\RequirePackage{fix-cm}
\documentclass[smallextended]{svjour3}       
\smartqed  
\usepackage[draft]{fixme}

\usepackage{graphicx}
\usepackage{algorithm}
\usepackage{algorithmicx}
\usepackage{tikz}
\usepackage{wrapfig}
\usepackage{stmaryrd}
\usepackage[utf8]{inputenc}

\usepackage{amsmath, amssymb}
\usepackage{amsfonts}
\usepackage{textcomp}

\newcommand\nulleins{[\nearrow]}
\newcommand\einsnull{[\searrow]}
\newcommand\nullnull{[\rightharpoondown]}
\newcommand\einseins{[\rightharpoonup]}
\newcommand{\postset}[1]{{{#1}^{\bullet}}}
\newcommand{\preset}[1]{{{}^{\bullet}{#1}}}

\def\nodes{\range{n}}
\def\B{\mathbb B}
\def\M{\mathbb M}
\def\DEF{\stackrel{\Delta}=}
\def\EQDEF{\stackrel{\Delta}\Longleftrightarrow}
\def\encode#1{\Tilde #1} 
\def\encodemp#1{\Tilde{\Tilde #1}} 

\def\dite#1#2{\xrightarrow[\mathrm{#1}]{#2}}
\def\ite#1{\dite{gen}{#1}}
\def\aite#1{\dite{async}{#1}}
\def\site#1{\dite{sync}{#1}}
\def\mpite#1{\dite{mpa}{#1}}
\def\areach#1{\aite{#1}\nolinebreak\negthickspace{}^*\,}
\def\mpreach#1{\mpite{#1}\nolinebreak\negthickspace{}^*\,}

\def\closure{\negthickspace{}^*\,}
\def\paite#1{\dite{atom}{#1}}
\def\psite#1{\dite{step}{#1}}
\def\pmite#1{\dite{mstep}{#1}}
\def\piite#1{\dite{istep}{#1}}

\def\UP{\!\!\uparrow}
\def\DW{\!\!\downarrow}
\def\IF{\mathrm{\ if\ }}
\def\ELSE{\mathrm{\ else\ }}

\def\card#1{|#1|}
\def\range#1{\{1, \dots, #1\}}

\def\cpnt{(P,T,\mi{pre},\mi{cont},\mi{post},M_0)}

\def\bnofcpn#1{\llbracket #1 \rrbracket}
\def\cpnofbn#1{\llparenthesis #1 \rrparenthesis}

\def\f#1{\operatorname{#1}}
\def\DNF{\f{DNF}}

\def\nth{$^{\text{th}}$ }

\usepackage{color}

\tikzstyle{every matrix}=[ampersand replacement=\&]
\tikzstyle{shorthandoff}=[]
\tikzstyle{shorthandon}=[]
\def\BCube{
\matrix[column sep=0.8cm, row sep=1cm,font=\footnotesize] {
  \node (s010) {$010$}; \&
  \node (s110) {$110$}; \&
  \node (s011) {$011$}; \&
  \node (s111) {$111$};
\\
  \node (s000) {$000$}; \&
  \node (s100) {$100$}; \&
  \node (s001) {$001$}; \&
  \node (s101) {$101$};
\\
};
}
\tikzstyle{tplace}=[circle,draw,inner sep=1.5mm]
\tikzstyle{transition} = [rectangle, draw, inner sep=0,text width=0.4cm, minimum height=0.4cm, text centered]
\tikzstyle{state} = [circle, draw, inner sep=0,text width=0.4cm, minimum height=0.4cm, text centered]
\tikzstyle{event} = [rectangle, draw, inner sep=0,text width=0.7cm,
minimum height=0.5cm, text centered,outer sep=0]
\tikzstyle{cond} = [circle, draw, inner sep=0,text width=0.4cm,
minimum height=0.4cm, text centered]
\tikzstyle{none} = [draw=none, fill=none, inner sep=0,minimum height=0.2cm]

\newcommand{\eqdef}{\stackrel{\Delta}{=}}
\newcommand{\setif}{\stackrel{\Delta}{\Longleftrightarrow}}

\newcommand{\varn}{\mathsf{v}}
\newcommand{\vari}{\mathit{var}}
\newcommand{\vali}{\mathit{val}}

\newcommand\by{\begin{eqnarray}}
\newcommand\ey{\end{eqnarray}}
\newcommand\bys{\begin{eqnarray*}}
\newcommand\eys{\end{eqnarray*}}

\newcommand\bet{\begin{theorem}}
\newcommand\ent{\end{theorem}}
\newcommand\bel{\begin{lemma}}
\newcommand\enl{\end{lemma}}
\newcommand\bec{\begin{corollary}}
\newcommand\enc{\end{corollary}}
\newcommand\bum{\begin{enumerate}}
\newcommand\eum{\end{enumerate}}
\newcommand\bit{\begin{itemize}}
\newcommand\eit{\end{itemize}}

\newcommand\stepn{\tau}
\newcommand{\markn}{\mathit{M}}

\newcommand{\non}[1]{\overline{#1}} 
\newcommand{\context}[1]{{\underline{#1}}}

\newcommand{\imasycf}{\leadsto}

\newcommand{\places}{\mathit{P}}
\newcommand\placen{\mathit{p}}
\newcommand\transn{\mathit{t}}

\newcommand\trans{{\mathit{T}}}
\newcommand\netn{N}

\newcommand{\mi}[1]{\mathit{#1}}
\newcommand{\pre}[1]{{^{\bullet}{#1}}}
\newcommand{\post}[1]{{{#1}^{\bullet}}}
\newcommand{\cont}[1]{\underline{#1}}
\newcommand{\precont}[1]{\pre{\cont{#1}}}

\newcommand{\splitting}{\mi{split}}
\newcommand{\Approx}{\mi{Approx}}
\newcommand{\abstr}{\mi{abstr}}

\newcommand\upt{\mathbf{up}}
\newcommand\downt{\mathbf{dw}}

\begin{document}
\sloppy

\title{Concurrency in Boolean networks}


\author{Thomas Chatain \and
    Stefan Haar \and
    Juraj Kol\v c\'ak \and
    Lo\"ic Paulev\'e \and Aalok Thakkar
}

\authorrunning{Chatain, Haar, Kol\v c\'ak, Paulev\'e, Thakkar}

\institute{T. Chatain, S. Haar, J. Kol\v c\'ak, A. Thakkar\at
    Inria Saclay-\^Ile-de-France\\
    LSV, CNRS, \& ENS Paris-Saclay\\
    Universit\'e Paris-Saclay, France\\
  \email{thomas.chatain@lsv.fr}\\
  \email{stefan.haar@lsv.fr}\\
  \email{juraj.kolcak@lsv.fr}\\
 \emph{Present address:} of A. Thakkar\at
 Department of Computer and Information Science\\
 University of Pennsylvania\\
Philadelphia, PA 19104, USA\\
   \email{athakkar@seas.upenn.edu}\\
           \and
L. Paulev\'e \at
Univ. Bordeaux, Bordeaux INP, CNRS, LaBRI, UMR5800\\
F-33400 Talence, France\\
CNRS, LRI UMR 8623, Univ. Paris-Sud -- CNRS\\
Universit\'e Paris-Saclay, France\\
\email{loic.pauleve@labri.fr}
}

\date{Received: date / Accepted: date}

\maketitle

\begin{abstract}
Boolean networks (BNs) are widely used to model the qualitative dynamics of biological systems.
Besides the logical rules determining the evolution of each component with respect to the state of
its regulators, the scheduling of component updates can have a dramatic impact on the predicted
behaviours.
In this paper, we explore the use of Read (contextual) Petri Nets (RPNs) to study dynamics of BNs from a
concurrency theory perspective.
After showing bi-directional translations between RPNs and BNs and analogies between results on synchronism
sensitivity, we illustrate that usual updating modes for BNs can miss plausible behaviours, i.e.,
incorrectly conclude on the absence/impossibility of reaching specific configurations.
We propose an encoding of BNs capitalizing on the RPN semantics enabling more behaviour than the generalized asynchronous updating mode.
The proposed encoding ensures a correct abstraction of any multivalued refinement, as
one may expect to achieve when modelling biological systems with no assumption on its time features.

\end{abstract}

\section{Introduction}
\label{sec:introduction}

Boolean networks (BNs) model dynamics of systems where several components (nodes) interact.
They specify for each node an update function to determine
its next value according to the configuration (global state) of the network.
In addition, an \emph{update mode} for scheduling the application of functions has to be
specified to determine the set of reachable configurations.

BNs are increasingly used to model dynamics of biological interaction networks, such as gene networks and
cellular signalling pathways.
In these practical applications, it is  usual to assess the accordance of a BN with the concrete
modeled system by checking if the observed behaviours are reproducible by the abstract BN
\cite{Rougny2016,Traynard16,Collombet2017}.
For instance, if one observes that the system can reach a configuration $y$ from
configuration $x$, one may expect it is indeed the case in the BN model.
The designed Boolean functions typically do not model the system correctly
whenever it is not the case and should thus be fixed prior to further analysis.
With this perspective, the choice of the update mode is crucial, as it is known
to have a strong influence on the reachable configurations of the network.

More fundamentally, the relationships between different updating modes have been
extensively studied for function-centered models such as
cellular automata
\cite{Schonfisch1999,Baetens2012} and
Boolean networks
\cite{Kauffman69,Thomas73,Garg08-Bioinf,Aracena09,Noual2017,Aracena16},
on which this article is focused.

Interestingly, the study of updating mechanisms in networks and their effect on
the emerging global dynamics has also been widely addressed in the field of
discrete and hybrid concurrent systems, especially with Petri
nets \cite{JanKou97,bald-corr-mont,BusiP96,vogler-TCS02,Winkowski98}.
Petri nets are a classical formal framework for studying concurrency, offering a
fine-grained specification of the \emph{conditions} (partial configurations) for
\emph{events} (partial configuration changes).
This decomposed view of causality and effect of updates enables capturing events
which can indifferently occur sequentially or in parallel, and events having conflicts
(triggering one would pre-empt the application of the second).

In the literature, many variants of Petri nets have been employed to model and
simulate various biological processes (see \cite{Goss98,Popova-Zeugmann05} for examples and
\cite{DBLP:journals/bib/Chaouiya07} for a review paper),
but little work considered
the link between the theoretical work on concurrency in Petri nets and the
theoretical work in Boolean networks.
In \cite{Steggles07,ChaouiyaNRT11,CHJPS14-CMSB}, encodings of BNs and their multi-valued
extension in certain classes of Petri nets have been proposed, often as means to
take advantage of existing dynamical analysis already implemented for Petri
nets, e.g., model-checking.

This paper aims at building a bridge between the theoretical work in BNs on the one hand and Read Petri Nets
(RPNs), also known as contextual Petri nets, on the other.
RPNs augment ordinary Petri Nets (PNs) with read arcs to model read-only access to resources.
It is always possible to simulate a RPN by an ordinary PN,
see Figure \ref{fig:redvsord} and the discussion below. 
Our  choice of using RPN is motivated by the fact that the connection between BNs and RPNs is more intuitive; but there is also an important technical advantage in using RPNs directly, rather than equivalent ordinary models.
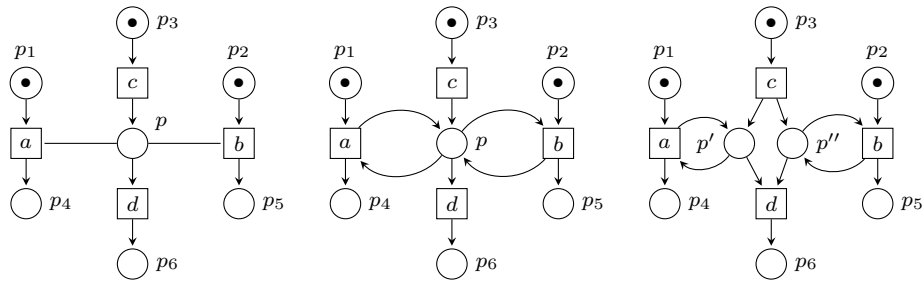
\begin{figure}
\def\b{1.4}
\def\c{.8}

\begin{tikzpicture}[>=stealth,shorten >=1pt,node distance=\c cm,auto]
  \node[state] (p_1a) at (3*\b, 0*\c) [label=above:\text{$p_1$}] {$\bullet$}; 
  \node[state] (p_1b) at (6*\b, 0*\c) [label=above:\text{$p_1$}] {$\bullet$}; 
  \node[state] (p_1) at (0*\b, 0*\c) [label=above:\text{$p_1$}] {$\bullet$};
  \node[state] (p_2) at (2*\b, 0*\c) [label=above:\text{$p_2$}] {$\bullet$};
  \node[state] (p_2a) at (5*\b, 0*\c) [label=above:\text{$p_2$}] {$\bullet$};
  \node[state] (p_2b) at (8*\b, 0*\c) [label=above:\text{$p_2$}] {$\bullet$};
  \node[state] (p_3) at (1*\b, 1*\c) [label=right:\text{$p_3$}] {$\bullet$};
  \node[state] (p_3a) at (4*\b, 1*\c) [label=right:\text{$p_3$}] {$\bullet$};
  \node[state] (p_3b) at (7*\b, 1*\c) [label=right:\text{$p_3$}] {$\bullet$};
  \node[state] (p_4) at (0*\b, -2*\c) [label=right:\text{$p_4$}] {}; 
 \node[state] (p_4a) at (3*\b, -2*\c) [label=right:\text{$p_4$}] {}; 
 \node[state] (p_4b) at (6*\b, -2*\c) [label=right:\text{$p_4$}] {}; 
  \node[state] (p) at (1*\b, -1*\c) [label=20:\text{$p$}] {}; 
  \node[state] (pa) at (4*\b, -1*\c) [label=right:\text{$p$}] {};
   \node[state] (pb1) at (6.7*\b, -1*\c) [label=left:\text{$p'$}] {}; 
   \node[state] (pb2) at (7.2*\b, -1*\c) [label=right:\text{$p''$}] {};
  \node[state] (p_5a) at (5*\b, -2*\c) [label=right:\text{$p_5$}] {}; 
  \node[state] (p_5b) at (8*\b, -2*\c) [label=right:\text{$p_5$}] {}; 
  \node[state] (p_5) at (2*\b, -2*\c) [label=right:\text{$p_5$}] {};
  \node[state] (p_6) at (1*\b, -3*\c) [label=right:\text{$p_6$}] {};
 \node[state] (p_6a) at (4*\b, -3*\c) [label=right:\text{$p_6$}] {};
 \node[state] (p_6b) at (7*\b, -3*\c) [label=right:\text{$p_6$}] {};

  \node[transition] (a) at (0*\b, -1*\c) {$a$};
  \path[->] (p_1) edge (a);
  \path[-] (p) edge (a);
  \path[->] (a) edge (p_4);

  \node[transition] (aa) at (3*\b, -1*\c) {$a$};
  \path[->] (p_1a) edge (aa);
  \path[->] (pa) edge  [bend left=50] (aa);
  \path[->] (aa) edge  [bend left=50] (pa);
  \path[->] (aa) edge (p_4a);

  \node[transition] (ab) at (6*\b, -1*\c) {$a$};
  \path[->] (p_1b) edge (ab);
  \path[->] (pb1) edge [bend left=50] (ab);
  \path[->] (ab) edge [bend left=50] (pb1) ;
  \path[->] (ab) edge (p_4b);

  \node[transition] (b) at (2*\b, -1*\c) {$b$};
  \path[->] (p_2) edge (b);
  \path[-] (p) edge (b);
  \path[->] (b) edge (p_5);
  
  \node[transition] (ba) at (5*\b, -1*\c) {$b$};
  \path[->] (p_2a) edge (ba);
  \path[->] (pa) edge  [bend left=50] (ba);
  \path[->] (ba) edge  [bend left=50] (pa);
  \path[->] (ba) edge (p_5a);
  
  \node[transition] (bb) at (8*\b, -1*\c) {$b$};
  \path[->] (p_2b) edge (bb);
  \path[->] (pb2) edge [bend left=50] (bb);
  \path[->]  (bb) edge [bend left=50] (pb2);
  \path[->] (bb) edge (p_5b);

  \node[transition] (c) at (1*\b, 0*\c) {$c$};
  \path[->] (p_3) edge (c);
  \path[->] (c) edge (p);
  
  \node[transition] (ca) at (4*\b, 0*\c) {$c$};
  \path[->] (p_3a) edge (ca);
  \path[->] (ca) edge (pa);
  
  \node[transition] (cb) at (7*\b, 0*\c) {$c$};
  \path[->] (p_3b) edge (cb);
  \path[->] (cb) edge (pb1);
 \path[->] (cb) edge (pb2);

  \node[transition] (d) at (1*\b, -2*\c) {$d$};
  \path[->] (p) edge (d);
  \path[->] (d) edge (p_6);
  
  \node[transition] (da) at (4*\b, -2*\c) {$d$};
  \path[->] (pa) edge (da);
  \path[->] (da) edge (p_6a);
  
  \node[transition] (db) at (7*\b, -2*\c) {$d$};
  \path[->] (pb1) edge (db);
  \path[->] (pb2) edge (db);
  \path[->] (db) edge (p_6b);


\end{tikzpicture}
\caption{\label{fig:redvsord} A Read Petri net $\mathcal{R}$ (left) and two different interpretations $\mathcal{N}_1$ and $ \mathcal{N}_2$ (center, right) of $\mathcal{R}$  as ordinary Petri nets; following \cite{BBCKRS12}.}
\end{figure}

Let us examine Figure \ref{fig:redvsord} more closely. In the read net $\mathcal{R}$ on the  left hand side, transitions $a$ and $b$ will be enabled while $p$ is marked, i.e. between the firings of $c$ and $d$; once a token is available on $p$, both $a$ and $b$ can fire independently and jointly, because the read arcs linking them to $p$ do not require \emph{removal} of the token from $p$. In $\mathcal N_1$ in the middle,  the firing of $a$ and $b$ in any order is still possible, however their \emph{synchronous} firing is prohibited by the conflict over the token on $p$. Only Petri net $\mathcal N_2$ is equivalent to $\mathcal{R}$; synchronous firing of $a$ and $b$ is obtained at the expense of duplicating $p$ by creating $p'$ and $p''$. In other words, faithful rendering of read net behaviour by ordinary nets requires the creation of multiple places for each place 'read' by several transitions, in order to pass from a read net to an ordinary net model. As will be seen, the constructions required for translation between read nets and BN in their turn also multiply place elements ; putting these constructions together is possible, but makes the resulting nets still larger and a lot less intuitive to apprehend and analyze.

In this paper, we consider the class of \emph{safe} (or 1-bounded) RPNs where each place can be marked
by at most one token, which makes it a natural choice for linking with Boolean networks.
This class has been extensively studied in the literature
and enables fine-grained definitions of different concurrent semantics
as it is detailed in Sect.~\ref{sec:CPN} and on which results of this article are built upon.

Below, we will give 
 bi-directional equivalent connections between the two formalisms of BNs and RPNs; this allows to  we use a
classical result from Petri net theory to show the PSPACE-completeness of reachability
in asynchronous BNs.
Then, we exhibit analogies of results on update mode comparisons.
Importantly, we show how the concurrent view of updates brings new
updating modes for BNs, enabling new behaviours and meeting with a correct
abstraction of multi-level systems. This result is illustrated on a small BN
which occurs in different models of actual biological networks, and for which the usual updating
modes fail to capture behaviours existing in refined models
(Sect.~\ref{sec:example}; Fig.~\ref{fig:example-beyond}).

\paragraph{Outline.}
Sect.~\ref{sec:bn} gives basic definitions of BNs, their asynchronous, synchronous,
and generalized asynchronous update mode, and their influence graph.
Sect.~\ref{sec:CPN} defines safe RPNs and their atomic, step, and interval semantics.
Sect.~\ref{sec:encodings} brings encodings of BNs into safe RPNs and vice-versa, the latter allowing
to derive that reachability in BNs is PSPACE-complete.
Sect.~\ref{sec:sensitivity} establishes an analogy between the results on synchronism sensitivity in
BNs and RPNs.
Sect.~\ref{sec:bn-is} provides an encoding of the interval semantics of RPNs into asynchronous BNs,
initially published in the conference paper \cite{beyond-general}.
Sect.~\ref{sec:beyond-general} first illustrates the benefits of the interval semantics on a simple BN
showing that usual BN semantics can miss plausible behaviours.
Then, an extension of the interval semantics is proposed in order to meet with a correct abstraction
of behaviours achievable in a multivalued refinement.
Finally, Sect.~\ref{sec:discuss} summarizes the contributions and discusses further work.

\paragraph{Notations.}
If $S$ is a finite set, $\card S$ denotes its cardinality.
$\B =\{0,1\}$,  and we write $\wedge$, $\vee$, $\neg$ for logic operators \emph{and},
\emph{or}, \emph{not};
given a set of literals $L=\{l_1, \dots, l_k\}$,
$\bigwedge_L \equiv l_1\wedge \dots\wedge l_k$ with
$\bigwedge_{\emptyset} = 1$, and
$\bigvee_L \equiv l_1\vee \dots\vee l_k$
with
$\bigvee_{\emptyset} = 0$.

\section{Boolean networks with function-centered specification}\label{sec:bn}

\label{sec:defs}
Given a \emph{configuration} \(x\in \B^n\) and \(i \in \range n\), we denote \(x_i\) the
\(i\textsuperscript{th}\) component of \(x\), so that \(x = x_1 \dots x_n\).
Given two configurations $x,y\in \B^n$, the components that differ are noted
$\Delta(x,y)\DEF\{ i\in \range n\mid x_i\neq y_i\}$.

\begin{definition}[Boolean network]
A Boolean network (BN) of dimension $n$ is a collection of functions
$f=\langle f_1, \ldots, f_n\rangle$ where
$\forall i\in\range n, f_i:\B^n\to \B$.
\end{definition}

Given $x\in \B^n$, we write $f(x)$ for $f_1(x)\dots f_n(x)$.

Fig.~\ref{fig:example}~(a) shows an example of BN of dimension 3.

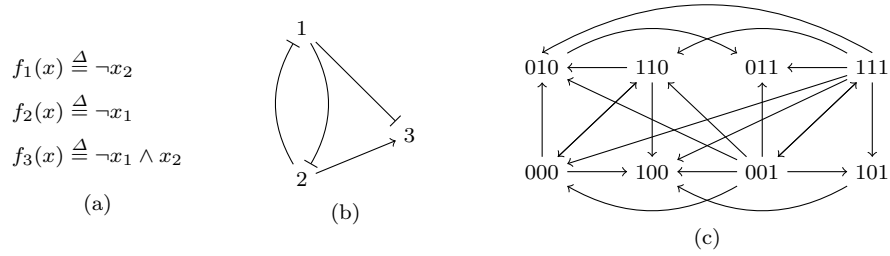
\begin{figure}[t]
\centering
\begin{minipage}{0.2\textwidth}
\centering
\begin{align*}
f_1(x) &\DEF \neg x_2\\
f_2(x) &\DEF\neg x_1\\
f_3(x) &\DEF \neg x_1 \wedge x_2
\end{align*}
(a)
\end{minipage}
\hfill
\begin{minipage}{0.3\textwidth}
\centering
\begin{tikzpicture}[node distance=2cm]
\node (1) {1};
\node[below right of=1] (3) {3};
\node[below of=1] (2) {2};
\path[->]
(1) edge[bend left,-|] (2)
(2) edge[bend left,-|] (1)
(1) edge[-|] (3)
(2) edge (3)
;

\end{tikzpicture}

(b)
\end{minipage}
\hfill
\begin{minipage}{0.45\textwidth}
\centering
\begin{tikzpicture}
\BCube
\path[->]
(s000) edge (s100)
(s000) edge (s110)
(s000) edge (s010)
(s010) edge[bend left] (s011)
(s110) edge (s100)
(s110) edge (s010)
(s110) edge (s000)
(s001) edge (s011)
(s001) edge (s110)
(s001) edge (s010)
(s001) edge (s100)
(s001) edge (s101)
(s001) edge[bend left] (s000)
(s001) edge (s111)
(s101) edge[bend left] (s100)
(s111) edge (s011)
(s111) edge[bend right] (s110)
(s111) edge (s100)
(s111) edge (s001)
(s111) edge (s101)
(s111) edge (s000)
(s111.north) edge[bend right] (s010.north)
;
\end{tikzpicture}

(c)
\end{minipage}
\caption{
(a) Example BN $f$ of dimension 3;
(b) Influence graph $G(f)$;
positive edges are with normal tip;
negative edges are with bar tip;
(c) Transition relations between states in $\B^n$ according to the generalized
asynchronous semantics of $f$.
}

\label{fig:example}
\end{figure}

When modelling biological systems, each node $i\in\nodes$ usually represents a biochemical species,
being either active (or present, value $1$) or inactive (or absent, value $0$).
Each function \(f_i\) indicates how is the evolution of the value of \(i\)
influenced by the current value of other components \(j \in \range n\).
However, this description can be interpreted in several ways, therefore several
updating modes coexist for BNs, depending on the assumptions about the order in which
the evolutions predicted by the \(f_i\) apply.

The (fully) \emph{asynchronous updating} assumes that only one component is
updated at each time step. The choice of the component to update is non-deterministic.
\begin{definition}[Asynchronous updating]
Given a BN $f$, the binary irreflexive relation $\aite f\,\subseteq \B^n\times\B^n$
is defined as:
\[
x\aite f y\EQDEF \exists i\in \range n, \Delta(x,y)=\{i\} \wedge y_i=f_i(x)
\enspace.
\]
We write $\areach f$ for the transitive closure of $\aite f$.
\end{definition}

The \emph{synchronous updating} can be seen as the opposite:
\emph{all} components are updated at each time step. This leads to a purely
deterministic dynamics.
\begin{definition}[Synchronous updating]
Given a BN $f$, the binary irreflexive relation $\site f\,\subseteq \B^n\times\B^n$
is defined as:
\[
x\site f y\EQDEF x\neq y\wedge\forall i\in\range n, y_i=f_i(x)
\enspace.
\]
\end{definition}

By forcing all the components to evolve simultaneously, the synchronous updating
makes a strong assumption on the dynamics of the system. In many concrete cases,
for instance in systems biology, this assumption is often unrealistic, at
least because the components model the quantity of some biochemical species
which evolve at different speeds.

As a result, the synchronous updating fails to
describe some behaviours, like the transition \(010 \rightarrow 011\) represented
in Fig.~\ref{fig:example}~(c) which represents the activation of species 3 when
species 1 is inactive and species 2 is active (\(f_3(010) = 1\)).
\medskip
There are also transitions which are possible in the synchronous but not in the
asynchronous updating, for instance \(000 \rightarrow 110\). Remark that
\(110\) is not even reachable from \(000\) in the asynchronous updating.
\medskip

The \emph{generalized asynchronous updating} generalizes both the asynchronous
and the synchronous updating: it allows updating synchronously any nonempty subset
of components.
\begin{definition}[Generalized asynchronous updating]
Given a BN $f$, the binary irreflexive relation $\ite f\,\subseteq \B^n\times\B^n$
is defined as:
\[
x\ite f y\EQDEF x\neq y\wedge \forall i\in\Delta(x,y): y_i=f_i(x)
\enspace.
\]
\end{definition}

Clearly, \(x\aite f y \Rightarrow x\ite f y\) and \(x\site f y \Rightarrow x\ite
f y\). The converse propositions are false in general.

\noindent
Note that we forbid “idle” transitions ($x\rightarrow x$) regardless of the
updating mode.

Other updating modes like sequential or block sequential have also been
considered in the literature on cellular automata and BNs
\cite{Baetens2012,Aracena09}, and usually lead to transitions allowed by the
generalized asynchronous updating.

\medskip

For each node $i\in\nodes$ of the BN, $f_i$ typically depends only on a subset of
nodes of the network.
The \emph{influence graph} of a BN (also called interaction or causal graph)
summarizes these dependencies by having an edge from node $j$ to $i$ if $f_i$
depends on the value of $j$.
Formally, $f_i$ depends on $x_j$ if there exists a configuration $x\in \B^n$ such that
$f_i(x)$ is different from $f_i(x')$ where $x'$ differs from $x$ solely in the
component $j$ ($x'_j = \lnot x_j$).
Moreover, assuming $x_j=0$ (therefore $x'_j=1$), we say that $j$ has a positive
influence on $i$ (in configuration $x$) if $f_i(x) < f_i(x')$, and a negative influence
if $f_i(x) > f_i(x')$.
It is possible that a node has different signs of influence on $i$ in different
configurations, leading to non-monotonic $f_i$.
It is worth noticing that different BNs can have the same influence graph.

\begin{definition}[Influence graph]
Given a BN $f$, its \emph{influence graph} $G(f)$ is a directed graph $(\range n,E_+,E_-)$ with
\emph{positive}
and \emph{negative} edges such that
\begin{align*}
(j,i)\in E_+ &\EQDEF
    \exists x,y\in \B^n: \Delta(x,y)=\{j\}, x_j<y_j, f_i(x) < f_i(y)
\\
(j,i)\in E_- &\EQDEF
    \exists x,y\in \B^n: \Delta(x,y)=\{j\}, x_j<y_j, f_i(x) > f_i(y)
\end{align*}

A (directed) cycle composed of edges in $E_+\cup E_-$ is said \emph{positive} when it is
composed by an even number of edges in $E_-$ (and any number of edges in $E_+$), otherwise it is \emph{negative}.

When $E_+\cap E_-=\emptyset$, we say that $f$ is \emph{locally monotonic}.
\end{definition}

The influence graph is an important object in the literature of BNs
\cite{TT95,ADG04a}.
For instance, many studies have shown that one can derive dynamical features of a BN $f$ by the sole
analysis of its influence graph $G(f)$.
Importantly, the presence of negative and positive cycles in the influence graph, and the way they
are intertwined can help to determine the nature of attractors
(that are the smallest sets of configurations closed by the transition relationship)
\cite{Richard10-AAM},
and derive bounds on the number of fixpoints and attractors a BN having the same
influence graph can have
\cite{RRT08,Aracena08,ARS17}.

\section{Read Petri Nets with transition-centered specifications}
\label{sec:CPN}

In the semantics of BNs, each node computes its next value
according to the value of the other nodes. We have seen in the previous section that
this general rule does not suffice to define the precise behaviour and several
updating modes can be considered.

This situation is very similar to what happens in contextual or Read Petri nets (RPNs),
where read arcs have been introduced to model read-only access to resources, for a matter of concurrency.
Interestingly, the introduction of read arcs in
Petri nets has also led to several variants of the semantics. In this section,
we present some of them, mainly taken from~\cite{CHKS-pn15}. Next, relying on a
natural encoding of BNs in RPNs (Sect.~\ref{sec:encodings}), we will establish
a correspondence between updating modes for BNs and semantics of RPNs. In
particular, we transpose the \emph{interval semantics} of RPNs to a new
semantics for BNs (Sect.~\ref{sec:bn-is}) which retrieves some plausible
scenarios that were missed by other updating modes.

\subsection{Read Petri Nets}
\label{sec:ConPN}

We consider only \emph{safe} Read Petri nets (RPNs), i.e., RPNs with at most one token in each place at any time.

\begin{figure}[tb]
  \centering
  \def\b{2}
\def\c{1.2}

\begin{tikzpicture}[>=stealth,shorten >=1pt,node distance=\c cm,auto]
  \node[state] (p_1) at (0*\b, 0*\c) [label=above:\text{$p_1$}] {$\bullet$};
  \node[state] (p_2) at (2*\b, 0*\c) [label=above:\text{$p_2$}] {$\bullet$};
  \node[state] (p_3) at (1*\b, -1*\c) [label=right:\text{$p_3$}] {$\bullet$};
  \node[state] (p_4) at (0*\b, -2*\c) [label=left:\text{$p_4$}] {};
  \node[state] (p_5) at (2*\b, -2*\c) [label=right:\text{$p_5$}] {};
  \node[state] (p_6) at (1*\b, -3*\c) [label=right:\text{$p_6$}] {};

  \node[transition] (a) at (0*\b, -1*\c) {$a$};
  \path[->] (p_1) edge (a);
  \path[-] (p_2) edge (a);
  \path[->] (a) edge (p_4);

  \node[transition] (b) at (2*\b, -1*\c) {$b$};
  \path[->] (p_2) edge (b);
  \path[-] (p_1) edge (b);
  \path[->] (b) edge (p_5);

  \node[transition] (c) at (1*\b, -2*\c) {$c$};
  \path[->] (p_3) edge (c);
  \path[-] (p_1) edge (c);
  \path[-] (p_5) edge (c);
  \path[->] (c) edge (p_6);

  \node[transition] (d) at (1*\b, -4*\c) {$d$};
  \path[->] (p_4) edge (d);
  \path[->] (p_5) edge (d);
  \path[->] (p_6) edge (d);
  \path[->] (d) edge [bend left=100] (p_1);
  \path[->] (d) edge [bend right=100] (p_2);
  \path[->] (d) edge [bend left=80] (p_3);

\end{tikzpicture}
  \caption{A Read Petri net (RPN). Neither atomic semantics nor step semantics
    allow \(d\) to fire, while the more permissive non-atomic semantics allows
    it.}
  \label{fig:interval}
\end{figure}
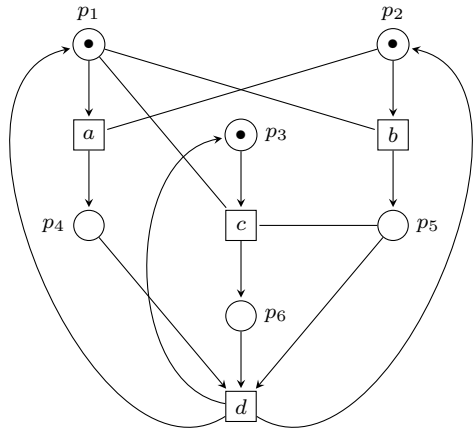

\begin{definition}[Read Petri Net (RPN)]
  A {\em Read Petri net} is a tuple \((P, T, \mi{pre}, \mi{cont},
  \mi{post}, M_0)\) where $P$ and $T$ are finite sets of {\em places} and {\em
    transitions} respectively, $\mi{pre}$, $\mi{cont}$ and $\mi{post}$ map each transition
  \(t \in T\) to its (nonempty) \emph{preset} denoted \(\pre{t} \eqdef \mi{pre}(t)
  \subseteq P\), its (possibly empty) {\em context} denoted \(\cont{t} \eqdef
  \mi{cont}(t) \subseteq P \setminus \pre{t}\) and its (possibly empty) {\em
    postset} denoted \(\post{t} \eqdef \mi{post}(t) \subseteq P\); \(M_0
  \subseteq P\) is the \emph{initial marking}.
  We usually denote \(\precont{t} \eqdef \pre{t} \cup \cont{t}\).
\end{definition}
For simplicity, we assume that for every transition $t$, its context is
disjoint from its preset and postset.
\medskip

A RPN is represented as a graph with two types of nodes: places
(circles) and transitions (rectangles). Presets are represented by arrows from
places to transitions, postsets by arrows from transitions to places, and
contexts by undirected edges, called read arcs, between places and transitions.
The initial marking is represented by tokens in places.
Fig.~\ref{fig:interval} shows an example of RPN.
The transition $a$, for instance, has $p_1$ in its preset, $p_2$ in its context
and $p_4$ in its postset.

\subsection{Atomic Semantics}

A \emph{marking} of a safe RPN is a set \(M \subseteq P\) of
marked places.
A Petri net starts in its {\em initial marking} \(M_0\).
A transition $t \in T$ is {\em enabled} in a marking \(M\) if all the places of
its preset and context are marked, i.e., \(\pre{\cont{t}} \subseteq M\).
Then $t$ can {\em fire} from \(M\), leading to the marking
\(M' \eqdef (M \setminus \pre{t}) \cup \post{t}\).
In this case, we write $M\paite{\netn, t}M'$ or simply $M\paite{\netn}M'$.

As we consider only \emph{safe} RPNs, we assume
that if a transition $t \in T$ is enabled in a marking \(M\), then \((M
\setminus \pre{t}) \cap \post{t} = \emptyset\).

\begin{definition}[Atomic semantics, a-run]
  We call {\em firing sequence of \(\netn\) under the atomic semantics}, or
  \emph{a-run}, any sequence \(\sigma \eqdef (t_1 \dots t_n)\) of transitions
  for which there exist markings \(M_1, \dots, M_n\) such that for all \(i \in
  \{1, \dots, n\}\), firing $t_i$ from $M_{i - 1}$ is possible and leads to
  $M_i$.
\end{definition}


For instance, the net in Fig.~\ref{fig:interval} has two possible firing
sequences: $(a)$ and $(bc)$. However, $d$ can never fire because
that would require to fire both $a$ and $b$ first, and firing
one of $a,b$ disables the other.

\subsection{Non-atomic Semantics}
\label{sec:non-atomic}

In this section, we discuss two semantics for concurrent firing of multiple
transitions. One is the well-known \emph{step semantics}
\cite{DBLP:journals/tcs/JanickiK93}, in which multiple
transitions can fire simultaneously. This is typically the case of \(a\) and
\(b\) in the net of Fig.~\ref{fig:interval}, which are both enabled
and have disjoint presets, but cannot fire together according to the atomic
semantics.
The step semantics can be interpreted as first checking whether all members of a set of transitions
can fire, and then firing them simultaneously.
Intuitively, the step semantics is somehow similar to the general asynchronous updating as it
considers any set of fireable transitions; whereas the maximal step semantics which considers only
maximal sets of fireable transitions is analoguous to the synchronous updating.
We then recall the \emph{interval semantics} introduced in \cite{CHKS-pn15}, which allows a more
liberal choice of checking and firing transitions in a set.

We present the semantics under the assumption that
the underlying net is safe even under these two semantics, which allow more
possibilities than the atomic one.

\subsubsection{Step Semantics}
\label{sec:steps}

We first recall the step semantics \cite{DBLP:journals/tcs/JanickiK93}.

\begin{definition}[Step semantics, s-run]
  Let $N$ be a RPN.
  A \emph{step} is a \emph{set} \(S\) of transitions of \(N\). It can fire from
  configuration \(M\) and lead to configuration \(M'\), written \(M
  \psite{\netn, S} M'\) or simply \(M \psite{\netn} M'\), if
  \begin{itemize}
  \item every $t \in S$ is enabled in $M$,
  \item the presets of the transitions in \(S\) are disjoint, and
  \item \(M' = \left(M \setminus \bigcup_{t \in S}\pre{t}\right)
    \cup \bigcup_{t \in S}\post{t}\).
  \end{itemize}
  We call {\em s-run} of \(N\) any sequence \(\sigma \eqdef (S_1 \dots S_n)\) of
  \emph{steps} for which there exist markings \(M_1, \dots, M_n\) such that for
  all \(i \in \{1, \dots, n\}\), step \(S_i\) can fire from $M_{i - 1}$ and
  leads to $M_i$.
\end{definition}

A variant of step semantics, called maximal step semantics has received interest
in the literature~\cite{DBLP:journals/tcs/JanickiLKD86,DBLP:conf/forte/CourtiatS94}.
\begin{definition}[Maximal step semantics]
  The firing rule for the maximal step semantics is defined as \(M \pmite{\netn,
    S} M'\) (or simply \(M \pmite{\netn} M'\)) iff \(M \psite{\netn, S} M'\) and
  no larger step \(S' \supsetneq S\) can fire from \(M\).
\end{definition}

In the example of Fig.~\ref{fig:interval}, the step semantics allows one to
fire \(a\) and \(b\) in one step since they are both enabled in the initial
marking and \(\pre{a} \cap \pre{b} = \emptyset\). This gives the s-run \((\{a,
b\})\) in addition to the others which were already possible under the atomic
semantics; for instance the a-run involving \(b\) followed by \(c\), denoted \((b
c)\) for the atomic semantics, is simply rewritten as the s-run \((\{b\}
\{c\})\) under the step semantics. However, transition $d$ remains dead
since none of these s-runs contains all of $a$, $b$, and $c$.

The intuitive model underlying the step semantics is that
all the transitions in the step can first check, in any order, whether they
are enabled and not in conflict with one another, i.e., their presets are disjoint.
Once the checks have been performed, they can all fire, again in any order. Put differently, if we
denote
the checking phase of a transition $t$ by $t^-$ and its firing phase by $t^+$,
then every step consists of any permutation of the actions of type $t^-$
(for all transitions $t$ in the step), followed by any permutation of
the actions $t^+$. The notion introduced in Def.~\ref{def:split}
formalizes this intuition.

\begin{definition}[
    s\textsuperscript{$\pm$}-run]
\label{def:split}
  For every s-run \((T_1 \dots T_n)\) of a RPN \(N\), every
  concatenation \(u_1^-.u_1^+.\cdots.u_n^-.u_n^+\) of sequences \(u_i^-\) and
  \(u_i^+\), is an \emph{s\textsuperscript{$\pm$}-run} of \(N\), where every
  \(u_i^-\) is a permutation of the set \(\{t^- \mid t \in T_i\}\) and every
  \(u_i^+\) is a permutation of the set \(\{t^+ \mid t \in T_i\}\) (where
  \(T_i\) is a set of transitions of \(N\)).
\end{definition}

For example, the s-run \((\{b\} \{c\})\) yields the s\textsuperscript{$\pm$}-run
\((b^- b^+ c^- c^+)\) and the s-run \((\{a, b\})\) yields four
s\textsuperscript{$\pm$}-runs: \((a^- b^- a^+ b^+)\), \((a^- b^- b^+ a^+)\),
\((b^- a^- a^+ b^+)\) and \((b^- a^- b^+ a^+)\).

\subsubsection{Splitting Transitions for Understanding Steps}

Def.~\ref{def:split} formalizes a semantics of RPNs in which
the firing of a transition does not happen directly, but in two steps,
the checking of the pre-conditions and the actual execution.
In this section, we generalize this idea.

The left-hand side of Fig.~\ref{fig:splitting} shows a part of
the net in Fig.~\ref{fig:interval}, which consists of transition
$a$ with its preset $\{p_1\}$, context $\{p_2\}$,
and postset $\{p_4\}$. The construction on the right-hand side of
\ref{fig:splitting} illustrates the idea of splitting firing transitions
into two phases:
\begin{itemize}
\item every transition \(t\) is split into \(t^-\) and \(t^+\);
\item every place \(p\) is duplicated to \(p^c\) (meaning token in \(p\)
  available for consumption) and \(p^r\) (meaning token in \(p\) available
  for reading).
\end{itemize}
Similar ideas about splitting transitions can be found in several works, for
instance in \cite{DBLP:journals/ipl/Vogler95}.

Intuitively, if we apply this construction to all transitions from
Fig.~\ref{fig:interval}, then the s\textsuperscript{$\pm$}-runs of that
net correspond to a-runs of the newly constructed net. The following
Def.~\ref{def:splitting} provides the precise details of the
construction.

\begin{definition}[\boldmath\(\splitting(N)\)]
\label{def:splitting}
  Given a RPN \(N = (P, T, \mi{pre}, \mi{cont}, \mi{post},
  M_0)\), \(\splitting(N) \eqdef (P', T',
  \mi{pre}', \mi{cont}', \mi{post}', M_0')\) is the RPN where
  \begin{itemize}
  \item
    \(T'\) contains two copies, denoted \(t^-\) and \(t^+\) of every transition
    \(t \in T\).
  \item
    \(P'\) contains two copies, denoted \(p^c\) and \(p^r\) of every place
    \(p \in P\), plus one place \(p_t\) per transition \(t \in T\).
  \item \(\pre{t^-} \eqdef \{p^c \mid p \in \pre{t}\}\)
  \item \(\cont{t^-} \eqdef \{p^r \mid p \in \cont{t}\}\)
  \item \(\post{t^-} \eqdef \{p_t\}\)
  \item \(\pre{t^+} \eqdef \{p^r \mid p \in \pre{t}\} \cup \{p_t\}\)
  \item \(\cont{t^+} \eqdef \emptyset\)
  \item \(\post{t^+} \eqdef \{p^c \mid p \in \post{t}\} \cup \{p^r \mid p \in \post{t}\}\}\)
  \item \(M_0' \eqdef \{p^c \mid p \in M_0\} \cup \{p^r \mid p \in M_0\}\)
  \end{itemize}
\end{definition}

\begin{figure}[tb]
  \centering
  \def\b{1.2}
\def\c{0.6}
\def\delta{8}

\begin{tikzpicture}[>=stealth,shorten >=1pt,node distance=\c cm,auto]


  \node[state] (p1) at (0*\b, 0*\c) [label=above:\text{$p_1$}] {};
  \node[state] (p2) at (2*\b, 0*\c) [label=above:\text{$p_2$}] {};
  \node[state] (p3) at (0*\b, -6*\c) [label=below:\text{$p_4$}] {};

  \node[transition] (t) at (0*\b, -3*\c) {$a$};

  \path[->] (p1) edge (t);
  \path[-] (p2) edge [bend left=20] (t);
  \path[->] (t) edge (p3);


  \node[state] (p1c) at (-0.2*\b + \delta, 0*\c) [label=above:\text{$p_1^c$}] {};
  \node[state] (p1r) at (0.2*\b + \delta, 0*\c) [label=above:\text{$p_1^r$}] {};

  \node[state] (p2c) at (1.8*\b + \delta, 0*\c) [label=above:\text{$p_2^c$}] {};
  \node[state] (p2r) at (2.2*\b + \delta, 0*\c) [label=above:\text{$p_2^r$}] {};

  \node[state] (pt) at (0*\b + \delta, -3*\c) [label=left:\text{$p_a$}] {};

  \node[state] (p3c) at (-0.2*\b + \delta, -6*\c) [label=below:\text{$p_4^c$}] {};
  \node[state] (p3r) at (0.2*\b + \delta, -6*\c) [label=below:\text{$p_4^r$}] {};

  \node[transition] (tm) at (0*\b + \delta, -1.9*\c) {$a^-$};
  \node[transition] (tp) at (0*\b + \delta, -4.1*\c) {$a^+$};

  \path[->] (p1c) edge (tm);
  \path[-] (p2r) edge [bend left=15] (tm);
  \path[->] (tm) edge (pt);

  \path[->] (p1r) edge [bend left=30] (tp);
  \path[->] (pt) edge (tp);
  \path[->] (tp) edge (p3c);
  \path[->] (tp) edge (p3r);
\end{tikzpicture}
  \caption{The splitting of transition \(a\) (left) into \(a^-\) and
    \(a^+\) (right).}
  \label{fig:splitting}
\end{figure}
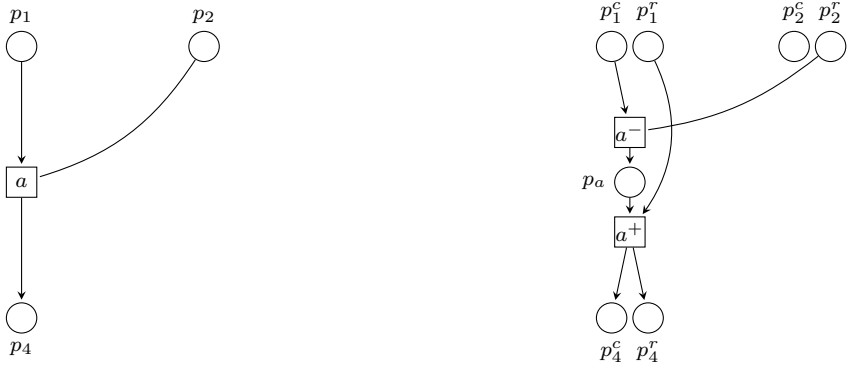

We now formally prove the intuition mentioned above:

\begin{lemma}\label{lem:s-run}
  Every s\textsuperscript{$\pm$}-run \(\sigma^{\pm}\) of \(N\) is an a-run of
  \(\splitting(N)\).
  Moreover \(\sigma^{\pm}\) reaches the marking \(\{p^c \mid p \in M\} \cup
  \{p^r \mid p \in M\}\), where \(M\) is the
  marking of \(N\) reached after the s-run \(\sigma\) from which
  \(\sigma^{\pm}\) is obtained.
\end{lemma}
\begin{proof}
  We proceed by induction on the length of \(\sigma\). The case \(\sigma = ()\)
  is trivial. Now, let \(\sigma^{\pm} = u_1^-.u_1^+.\cdots.u_n^-.u_n^+\) be an
  s\textsuperscript{$\pm$}-run obtained from an s-run \(\sigma = (T_1 \dots
  T_n)\), assume the property true for
  \(u_1^-.u_1^+.\cdots.u_{n-1}^-.u_{n-1}^+\) and denote \(M_{n-1}\) the marking
  reached after \((T_1 \dots T_{n-1})\). By induction hypothesis,
  \(u_1^-.u_1^+.\cdots.u_{n-1}^-.u_{n-1}^+\) reaches the marking \(\{p^c \mid p
  \in M_{n-1}\} \cup \{p^r \mid p \in M_{n-1}\}\)
  of \(\splitting(N)\). The fact that \(T_n\) is a valid step from \(M_{n-1}\)
  implies that \(\bigcup_{t \in T_n}\precont{t} \subseteq M_{n-1}\) and that the
  presets of the transitions in \(T_n\) are disjoint. This allows one to fire
  all the \(t^-\), \(t \in T_n\) in any order and reach the marking \(\{p^c \mid p
  \in M_{n-1} \setminus \bigcup_{t \in T_n}\pre{t}\} \cup \{p^r \mid p \in
  M_{n-1}\} \cup \{p_t \mid t
  \in T_n\}\) of \(\splitting(N)\). Now the \(t^+\), \(t \in T_n\), are all
  enabled and their presets are disjoint. They can in turn be fired in any
  order, reaching the desired marking of \(\splitting(N)\).\qed
\end{proof}

Note that the converse of Lemma~\ref{lem:s-run} does not hold. For instance, for
the net $N$ from Fig.~\ref{fig:interval}, the net $\splitting(N)$ admits the
a-run $a^-b^-b^+c^-c^+a^+$, which is not an s\textsuperscript{$\pm$}-run of $N$.

\subsubsection{Interval Semantics}
\label{sec:interval-semantics}

We have seen that the construction \(\splitting(N)\) admits firing sequences
that cannot be mapped back to executions under either the atomic or the
step semantics. In this section, we shall introduce the
\emph{interval semantics}, which is more general than the step semantics,
and whose interpretation on a net $N$ does correspond to the feasible
executions in $\splitting(N)$.

\begin{definition}[Interval semantics, i-run]
  Every a-run of \(\splitting(N)\) is called \emph{i-run} of \(N\), or run of
  \(N\) under the interval semantics.
\end{definition}

Coming back to the example of Fig.~\ref{fig:interval}, transition \(d\) can fire
under the interval semantics, for instance after the i-run \(a^- b^- b^+ c^- c^+
a^+ d^- d^+\) where transitions \(b\) and \(c\) complete the firing during the
period in which \(a\) fires.
Under the atomic semantics, \(a\) and \(b\) are in conflict, which prevents
\(d\) from firing. Under the step semantics, \(a\) and \(b\) can fire in the
same step, but then \(c\) cannot fire. Under the interval semantics,
$d$ can also fire.



Recall that we introduced $t^-$ and $t^+$ to represent different phases
during the execution of transition $t$. An obvious question is whether
the new semantics can lead to runs in which a transition `gets stuck'
during its execution. The following Lemma~\ref{lem:complete} affirms that
this is not the case: once $t^-$ is fired, nothing can hinder $t^+$ from
firing too.

\begin{definition}[complete marking]
  A marking of \(\splitting(N)\) is \emph{complete} if no \(p_t\) is marked.
\end{definition}
In particular, the initial marking is complete.

\begin{definition}[complete i-run]\label{def:complete-irun}
  An i-run is \emph{complete} if for each transition \(t^-\) in it, it includes the corresponding
  transition \(t^+\).
\end{definition}

\begin{lemma}
\label{lem:complete}
  Every i-run can be completed: for every i-run \(\sigma\), there exists a
  suffix \(\mu\) which matches all the unmatched \(t^-\), and such that
  \(\sigma\mu\) is an i-run.
  Moreover, complete i-runs (and only them) lead to complete markings.
\end{lemma}
\begin{proof}
  As long as a \(t^-\) is unmatched, \(\pre{t^+}\) remains included in the
  marking: no other transition consumes these tokens. Hence, it suffices to fire
  all the \(t^+\) corresponding to the unmatched \(t^-\), in any order.\qed
\end{proof}

Now, relating \(\splitting(N)\) with the original net \(N\), we map naturally
every marking \(M\) of \(N\) to the complete marking \(M'\) of \(\splitting(N)\)
defined as \(M' \eqdef \{p^c \mid p \in M\} \cup \{p^r \mid p \in M\}\). We get
of course that
\[M_1 \paite{\netn, t} M_2 \implies
M'_1 \paite{\netn, t^-} \paite{\netn, t^+} M'_2\,,\]
but in general the interval semantics induces more runs: for all markings
\(M_1\) and \(M_2\) of \(N\), we
write \(M_1\piite{\netn}\closure M_2\) when \(M'_1 \paite{\netn}\closure M'_2\).

\section{Encodings}
\label{sec:encodings}
\subsection{Coding Boolean Networks in safe Read Petri nets}
\label{sec:bn2cpn}

The translation of BNs into safe Petri nets has been addressed in the literature (e.g.
\cite{DBLP:journals/bib/Chaouiya07,ChaouiyaNRT11,CHJPS14-CMSB,DBLP:conf/apn/ChaouiyaRRT04}).
We provide here a similar encoding of BNs into safe RPNs, with the explicit specification of the context
of transitions, and with notations that will be used in Sect.~\ref{sec:sensitivity}.
The encoding can be easily generalized to multi-valued networks to safe RPNs, following
\cite{CHJPS14-CMSB,gored-TCBB}.

BNs translate into a special type of RPNs:
\bit
\item \emph{complemented}: for every place $\placen$ there is exactly one distinct place $\non{\placen}$ such that 
\bys
\preset{\placen}=\postset{\non{\placen}}~\land~\postset{\placen}=\preset{\non{\placen}}~\land~\forall \transn\in\trans:~\placen\in\context{\transn}\Rightarrow \non{\placen}\not\in\context{\transn}; 
\eys
\item \emph{Boolean}:
there is a surjection $\vari:\places\to\range n$ such that
\bys\forall p,p'\in\places:~
\vari(p)=\vari(p')&\Leftrightarrow& p'\in \{p,\non{p}\},
\eys 
and, subsequently, a mapping $\vali:\places\to\B$ which satisfies
\bys
\forall p\in\places:\vali(p)+\vali(\non p) = 1.
\eys
Moreover, any reachable marking $M$ satisfies
\bys
\forall p\in P: p\in M\Leftrightarrow \non p\notin M.
\eys
\item \emph{transition dichotomy}: every transition $\transn\in\trans$ has exactly one input place
$\placen$ and one output place $\bar{\placen}$. If $\vali(\placen)=0$ then call $\transn$ the \emph{up-transition} 
$\upt(\vari(\placen))$ of $\vari(\placen)$, otherwise the \emph{down-transition} $\downt(\vari(\placen))$ of $\vari(\placen)$.
\eit
Let us consider a BN $f$ of dimension $n$.
Each component $\varn\in\range n$ is modeled as two places $\varn_0$ and $\varn_1$ representing the two values
possible  for $\varn$.
Then a Petri net transition $\varn^+$ is defined
for each conjunctive clause of the disjunctive normal form of $(\neg x_\varn \wedge f_{\varn}(x))$.
Such a transition consumes a token in the place $\varn_0$ and produces a token in the place
$\varn_1$, and its context is formed by the places corresponding to the literals of the conjunction
other than $\neg x_{\varn}$:
for each component $\varn'\in\range n$, $\varn'\neq\varn$,
if the clause contains $x_{\varn'}$, the context contains the place ${\varn'_1}$;
if the clause contains $\neg x_{\varn'}$, the context contains the place ${\varn'_0}$.
A transition $\varn^-$ is defined similarly
for each conjunctive clause of the disjunctive normal form of $(x_\varn \wedge \neg f_{\varn}(x))$,
such that
\bys
\preset{(\varn^+)}=\postset{(\varn^-)}=\{\varn_0\} &\mathit{and}&
\preset{(\varn^-)}=\postset{(\varn^+)}=\{\varn_1\},
\eys

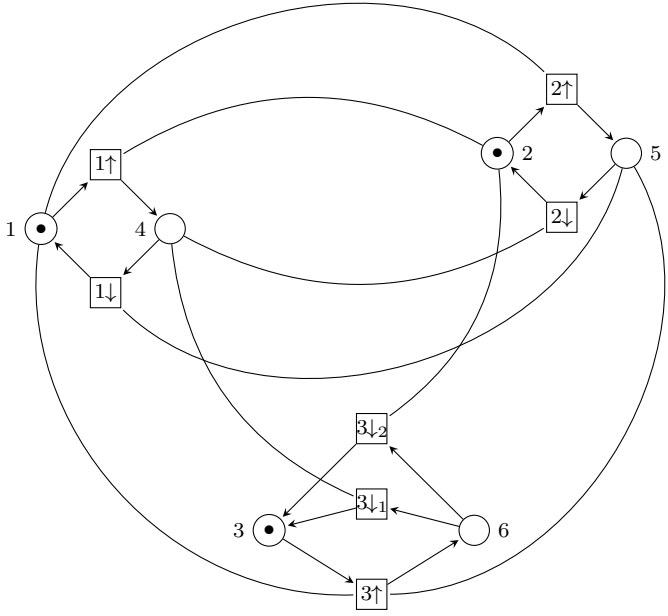
\begin{figure}[t]
\centering
\def\b{2}
\def\c{1.2}

\begin{tikzpicture}[>=stealth,shorten >=1pt,node distance=\c cm,auto]
\node[state] (p10) [label=left:$1$] {$\bullet$};
\node[transition, above right of=p10] (t1p) {$1\UP$};
\node[transition, below right of=p10] (t1m) {$1\DW$};
\node[state,below right of=t1p] (p11) [label=left:$4$] {};
\path[->] (p10) edge (t1p) (t1p) edge (p11) (p11) edge (t1m) (t1m) edge (p10);
\begin{scope}[xshift=6cm,yshift=1cm]
\node[state] (p20) [label=right:$2$] {$\bullet$};
\node[transition, above right of=p20] (t2p) {$2\UP$};
\node[transition, below right of=p20] (t2m) {$2\DW$};
\node[state,below right of=t2p] (p21) [label=right:$5$] {};
\end{scope}
\path[->] (p20) edge (t2p) (t2p) edge (p21) (p21) edge (t2m) (t2m) edge (p20);
\path[-]
    (p20) edge[bend right=30] (t1p)
    (p11) edge[bend right=30] (t2m)
    (p21) edge[bend left=60] (t1m)
    (p10) edge[bend left=60] (t2p)
;

\begin{scope}[xshift=3cm,yshift=-4cm]
\node[state] (p30) [label=left:$3$] {$\bullet$};
\node[transition,below right of=p30,xshift=5mm] (t3p) {$3\UP$};
\node[state,above right of=t3p,xshift=5mm] (p31) [label=right:$6$] {};
\node[transition,above right of=p30,xshift=5mm,yshift=-5mm] (t3m1) {$3\DW_1$};
\node[transition,above right of=p30,xshift=5mm,yshift=5mm] (t3m2) {$3\DW_2$};
\path[->] (p30) edge (t3p) (t3p) edge (p31)
 (p31) edge (t3m1) edge (t3m2)
 (t3m1) edge (p30)
 (t3m2) edge (p30)
;
\path[-]
    (p10) edge[bend right=50] (t3p)
    (p21) edge[bend left=60] (t3p)
    (p11) edge[bend right] (t3m1)
    (p20) edge[bend left] (t3m2)
 ;
\end{scope}

\end{tikzpicture}
\caption{RPN encoding of the BN of Fig.~\ref{fig:example}
$\langle f_1(x)=\neg x_2, f_2(x)=\neg x_1, f_3(x)=\neg x_1\wedge x_2\rangle$
and configuration $000$}
\label{fig:example2cpn}
\end{figure}

Fig.\ref{fig:example2cpn} shows the translation of the BN of Fig.~\ref{fig:example} into RPN.

Hereafter, Def.~\ref{def:bn2cpn} gives a formalization of this encoding, and
Theorem~\ref{thm:bn2cpn} states its correctness with respect to the asynchronous, synchronous,
generalized asynchronous updating modes, and
RPN atomic, maximal step, and step semantics, respectively.
Given a Boolean formula $F$, we write $\DNF[F]$ for the set of conjunctive clauses
in the disjunctive normal form of $F'$.
A clause $C\in\DNF[F]$ is then a set of literals, positives or negatives.
It is worth noticing that the resulting RPN can have a number of transitions exponential in the
number of literals in the Boolean functions.

\begin{definition}
Given a BN $f$ of dimension $n$ and a configuration $y$, $\cpnofbn f$ is the RPN $(P,T,\mi{pre},\mi{cont},\mi{post},M_0)$
such that
\begin{itemize}
\item $P = \{1, \dots, 2n\}$ are the places;
\item $T, \mi{pre},\mi{cont},\mi{post}$ are the smallest sets such that for each $i\in\range n$,
for each clause $C\in\DNF[\neg x_i\wedge f_i(x)]$ (resp. $C\in\DNF[x_i\wedge \neg f_i(x)]$),
there is a transition $t\in T$ such that
$\preset t = \{i\}$ (resp. $\preset t=\{i+n\}$), $\postset t=\{i+n\}$ (resp. $\postset t=\{i\}$),
and $\context t = \{ j \mid [\neg x_j]\in C, j\neq i\} \cup \{ j+n \mid [x_j]\in C, j\neq i\}$;
\item $M_0=\cpnofbn y$
\end{itemize}
where, for any configuration $x\in\B^n$,
$\cpnofbn x \DEF \{ i+nx_i \mid i \in \range n\}$
(e.g., 
$\cpnofbn {010} = \{ 1, 5, 3 \}$,
$\cpnofbn {101} = \{ 4, 2, 6 \}$).
\label{def:bn2cpn}
\end{definition}

\begin{theorem}
Given a BN $f$ of dimension $n$, for any configurations $x,y\in\B^n$,
\begin{align*}
x \aite{f} y &\Longleftrightarrow \cpnofbn x\paite{\cpnofbn f} \cpnofbn y\enspace,
\\
x \site{f} y &\Longleftrightarrow \cpnofbn x\pmite{\cpnofbn f} \cpnofbn y\enspace,
\\
x \ite{f} y &\Longleftrightarrow \cpnofbn x\psite{\cpnofbn f} \cpnofbn y\enspace.
\end{align*}
\label{thm:bn2cpn}
\end{theorem}
\begin{proof}
For any $i\in\range n$,
$f_i(x)\neq x_i$ if and only if there exists a transition $t$ of $\cpnofbn f$
where $\preset{\context t}\subseteq \cpnofbn x$.
\qed
\end{proof}

\subsection{Coding Read Petri Nets in Boolean Networks}

We have given above a translation of BNs into (a special class of) RPNs.
The comparison of both models also leads us into the opposite direction.

In the following, fix a safe RPN $\netn=\cpnt$.
The BN associated to $\netn$ has $\card\places+\card\trans$ components, where the first
$\card\places$ components encode the marking of the corresponding places,
and the $\card\trans$ other components encode the occurring transitions.
Without loss of generality, we assume that places and transitions range over indexes from $1$ to
$\card P+\card T$, i.e., $P \cup T \equiv \range{\card P+\card T}$.
In order to simplify the encodings, we additionally assume the RPNs
to be \emph{loop-free}, i.e., for every transition \(t \in T\), \(\pre{t} \cap
\post{t} = \emptyset\). It is well known that loops can be replaced by read
arcs without any effect on the (atomic) semantics.

Transporting the \emph{dynamics}, i.e., the actual \emph{firing} of transitions, into the framework
of BNs constitutes the non-trivial part of the translation.
A RPN transition typically has more than one output place, while the functions in
BNs write on one single variable.
Our encoding decomposes the firing of a RPN transition into several updates of the BN.
Essentially, when components corresponding to the pre-condition and context of a transition $t$ are marked,
and if no other transition $t'$ is already occurring,
the $t$\nth component of the BN can be updated to $1$.
Then, the components related to the input and output places of $t$ are updated (in any order) to apply
their respective un-marking and marking.
Once all these components have been updated, the $t$\nth component is updated to $0$.

It results that
a transition $t$ is occurring, encoded by the value $1$ of the $t$\nth component, if and only if
either (i) no transition is occurring, and all  components corresponding to places in the
pre-condition and context of $t$ have value $1$,
or (ii) $t$ is already occurring and at least one input (resp. output) place has not been unmarked
(resp. marked) yet.
A component corresponding to a place $p$ has value $1$ if and only if either one of transition
producing $p$ is occurring, or if it has already value $1$ and none transition consuming it is
occurring.


Hereafter, Def.~\ref{def:cpn2bn} provides a formalization of the encoding of a safe RPN into a
BN and Theorem~\ref{thm:cpn2bn} states its correctness in the scope of the asynchronous updating
(atomic);
note that the correctness also holds for the generalized asynchronous (step semantics) and
synchronous (maximal step semantics) updating.
\begin{definition}
Given a safe loop-free RPN $\netn = \cpnt$,
$\bnofcpn{\netn}$ is the BN of dimension $\card P+\card T$ such that
\begin{align*}
\forall p\in P,
\bnofcpn{\netn}_{p}(x) &=
    \left(\bigvee_{t\in \preset p} x_{t}\right)
    \vee
    \left(x_{p} \wedge \bigwedge_{t\in\postset p} \neg x_{t}\right)
\\
\forall t\in T,
\bnofcpn{\netn}_{t}(x) &=
    \left(\bigwedge_{p\in\preset {\context t}} x_{p}
    \wedge
    \bigwedge_{t'\in T} \neg x_{t'}\right)
\\
    &\vee
\left(
x_{t} \wedge 
\left(\bigvee_{p\in\postset t} \neg x_{p}
\vee
\bigvee_{p\in\preset t} x_{p}\right)
\right)
\end{align*}

Given a marking $M\subseteq P$ of $\netn$,
the corresponding configuration of $\bnofcpn{\netn}$ is
$\bnofcpn{M} \in \B^n$ where
$\forall p\in M, \bnofcpn{M}_{p} = 1$,
$\forall p\in P\setminus M, \bnofcpn{M}_{p} = 0$,
and
$\forall t\in T, \bnofcpn{M}_{t} = 0$.
\label{def:cpn2bn}
\end{definition}
As an example, let us consider the RPN of Fig.~\ref{fig:splitting}(left), which consists 3 places
$p_1$, $p_2$, $p_4$, and one transition $a$, such that
$\preset a = \{p_1\}$,
$\postset a =\{p_4\}$,
and $\context a = \{p_2\}$.
The above encoding into BN leads to 4 Boolean functions:
\begin{align*}
    f_{p_1}(x) &= x_{p_1} \wedge \neg x_a
\qquad
f_{p_2}(x) = x_{p_2}
\qquad
f_{p_4}(x) = x_a \vee x_{p_4}
\\
f_a(x) &= (x_{p_1} \wedge x_{p_2})
\vee (x_a \wedge (\neg x_{p_4} \vee x_{p_1}))
\end{align*}

\begin{theorem}
For a safe RPN $\netn=\cpnt$,
and any pair of markings $M,M'\subseteq P$, one has
\[
M \paite{\netn}\closure M'
\Longleftrightarrow
\bnofcpn{M} \areach{\bnofcpn{\netn}} \bnofcpn{M'}
\]
\label{thm:cpn2bn}
\end{theorem}
\begin{proof}
If $M=M'$, the proof is trivial; in the following we consider $M\neq M'$.

($\Rightarrow$)
Let us assume that $M\paite{\netn} M'$. Then there exists $t\in T$ such that $\preset{\context t} \subseteq M$
and $M' = (M \setminus \preset t) \cup \postset t$.
Thus, $\bnofcpn{\netn}_{t}(\bnofcpn{M}) = 1$, and therefore, there exists $y\in\B^n$ such that
$x \aite{\bnofcpn{\netn}} y$ with $\Delta(x,y)=\{t\}$.
Then, assuming $\preset t\cap \postset t=\emptyset$,
for each place $p\in\preset t$, because $t\in\postset p$ and $y_{p}=1$,
$\bnofcpn{\netn}_{t}(y)=0$, and
for each place $p\in\postset t$, because $t\in\preset p$,
$\bnofcpn{\netn}_{p} = 1$.
Therefore, by updating the components $p$ for $p\in\preset t\cup \postset t$ in any ordering, we
obtain a configuration $z$ where all components are $0$ except
the components $p, \forall p\in M'$, and the component $t$.
Then, because $\bnofcpn{\netn}_{t}(z)=0$, the latter component is set to $0$, resulting in
the configuration $\bnofcpn{M'}$.

($\Leftarrow$)
Let us assume there exists $y\in\B^{\card P+\card T}$ such that
$\bnofcpn{M}\aite{\bnofcpn{\netn}} y$.
Necessarily, there is a unique $t\in T$ such that $y_{t}=1$;
moreover, $\preset {\context t}\subseteq M$.
Remark that as long as the $t$\nth component of a configuration $x$ is $1$,
none of the other  components $t'$ for $t'\in T, t'\neq t$ can be set to $1$
(because $\bnofcpn{\netn}_{t'}(x)=0$).
Moreover, remark that in the configuration $y$, $\{ p \in P \mid y_{p}\neq \bnofcpn{\netn}_{p}\} =
\preset t\cup \postset t$, and that the component $t$ can be set to $0$ only when all these latter
components have been updated.
Therefore, with $M'' =(M\setminus \preset t)\cup\postset t$,
we obtain that
$\bnofcpn{M}\areach{\bnofcpn{\netn}} \bnofcpn{M''}$
and $M\paite{\netn} M''$.
\qed
\end{proof}

The reachability problem consists in deciding if there exists a sequence of transitions from a given
configuration (marking) $x$ to a given configuration $y$.
The reachability problem is PSPACE-complete in safe RPNs with asynchronous update mode~\cite{ChengEP95}.
By linear reduction to BNs, we therefore obtain that reachability in BNs is PSPACE-hard:
\begin{corollary}
Reachability in asynchronous BNs is PSPACE-hard.
\end{corollary}
Finally, one can remark that deciding the reachability in BNs is in PSPACE:
given a BN of dimension $n$ and the initial configuration $x$,
let us define a counter using $n$ bits, initially with value $0$.
Then, while the counter has value strictly less that $2^n$ and the current configuration is not equal to
$y$, non-deterministically apply an update, and increase the counter by one.
\begin{theorem}
Reachability in asynchronous BNs is PSPACE-complete.
\end{theorem}

\section{Synchronism sensitivity}
\label{sec:sensitivity}
For some 
BN or  RPN, changing the update/firing policy (from synchronous to asynchronous) may have little
impact on the reachable states. For others, it may render configurations reachable,  or exclude previously feasible paths. We say that a network of the latter category  is \emph{synchronism sensitive}. The authors of \cite{Noual2017}
have analyzed this sensitivity in BNs; in this section, we perform an analogous analysis for RPNs.
As we will show, the characterization of synchronism sensitivity in safe RPNs boils down to
the existence of \emph{preemption cycles}, defined below, among the  transitions that are enabled in a given marking.
Moreover, we show that when instantiated on RPNs encoding of BNs (according to
Sect.~\ref{sec:bn2cpn}),
the general characterization of synchronism sensitivity in RPNs allows to recover the results of
synchronism sensitivity in BNs with respect to their influence graph \cite{Noual2017}, with a slight
generalization relaxing the local monotonicity constraints of BNs.

\subsection{Synchronism sensitivity in BNs}
Following \cite{Noual2017}, given a BN $f$ of dimension $n$ where, $\forall i\in \range n$, $f_i$ is
monotonoic, a positive (resp. negative) edge $(j,i)$ of its influence graph $G(f)$
is \emph{frustrated} in a configuration $x\in\B^n$ iff $x_i\neq x_j$ (resp. $x_i=x_j$).
A (directed) cycle in $G(f)$ is \emph{critical} in $x$
iff all its edges are frustrated.

Then, the synchronism sensitivity in BNs can be characterized with respect to their influence graphs
as follows.
\bel[\cite{Noual2017}, Prop. 1] A critical cycle must 
be \emph{NOPE}: \textbf{n}egative with \textbf{o}dd length or \textbf{p}ositive with \textbf{e}ven length. 
\enl
\bet[\cite{Noual2017}] \label{thm:nope-bn}
Synchronism-sensitivity, i.e., the presence of some synchronous transition
that cannot be sequentialized,  in a locally monotonic BN $f$ requires the existence of a critical
cycle, and thus of a \emph{NOPE}-cycle in its influence graph $G(f)$.
\ent

\subsection{Synchronism Sensitivity in RPNs}

Given any safe RPN $\netn=\cpnt$,
call a pair $(\stepn,\markn)\in 2^\trans\times 2^\places$ such that
$\stepn$ is s-enabled but not a-enabled in $\markn$ a \emph{witness of synchronism sensitivity} or,
following \cite{Noual2017}, \emph{normal}.

As in \cite{bald-corr-mont}, we say for any two transitions $\transn_1,\transn_2\in \trans$ that $\transn_1$ 
\emph{preempts}\footnote {for readers familiar with \cite{bald-corr-mont}: we will only need this \emph{immediate} preemption relation $\imasycf$ here, not the full asymmetric conflict  obtained by adding causal precedence} $\transn_2$, written $\transn_1\imasycf \transn_2$ iff the context of $\transn_2$ intersects the preset of $\transn_1$:
\bys
\transn_1\imasycf \transn_2 &\setif&
\preset{\transn_1}\cap\context{\transn_2}.
\eys
\bet\label{th:cycles}
Let $(\stepn,\markn)\in 2^\trans\times2^\places$ such that $\markn$ s-enables $\stepn$. 
\bum\item If $\stepn=\{\transn_{1},\ldots,\transn_{n}\}$ is a \emph{preemption cycle}, i.e.,
\bys
\transn_{1}\imasycf\transn_{2}\imasycf\ldots \transn_{n-1}\imasycf\transn_{n},
\eys
then $(\stepn,\markn)$ is normal.
\item Conversely, if $(\stepn,\markn)$ is normal, then $\stepn$ contains a preemption cycle.
\eum
\ent
\begin{proof}
  Part 1 follows immediately from the assumptions. For Part 2, take any transition $\transn_1\in\stepn$. If there is no place $\placen\in\context{\transn}$ 
  such that $\placen\in\preset{\transn_2}$ for some $\transn_2\in\stepn$, remove $\transn_1$ from $\stepn$ and start over. Otherwise, we have 
  $\transn_2\imasycf\transn_1$, and inspect $\postset(\context{\transn_2})$ as above. Since $|\stepn|=n$, this process terminates after at most $n$ steps, yielding either a 
  decomposition of $\stepn$, or a preemption chain of length at most $n$,  or a preemption cycle of length at most $n$. Only the last case corresponds to $\stepn$ being normal.
\qed
\end{proof}
As an immediate consequence, we note the following minimality result:
\bec
Let $\stepn$ be such that $(\stepn,\markn)$ is normal, but every
$\emptyset\subset\stepn'\subseteq\stepn$ 
(with proper inclusions) is a-enabled, i.e., $(\stepn',\markn)$  is 
not normal. Then $\stepn$ is a minimal preemption cycle.
\enc

In Fig.~\ref{fig:ContextAneg}, $\tau=\{1\DW,2\DW,3\DW\}$ illustrates a preemption cycle, which is
also normal in the marking shown; $\tau'=\{1\UP,2\UP,3\UP\}$ is another preemption cycle which is
not enabled, but would become enabled after firing $\tau$.
In Fig.~\ref{fig:ContextB}, $\tau''=\{1\UP,2\DW\}$ is a preemption cycle, which is normal in the
marking shown.
\begin{figure}\begin{center}
\def\b{1.7}
\def\c{1.1}

\begin{tikzpicture}[>=stealth,shorten >=1pt,node distance=\c cm,auto]
  \node[state] (a_1) at (1*\b, -1.5*\c) [label=below:\text{$1_1$}] {$\bullet$};
  \node[state] (a_0) at (1*\b, -3.5*\c) [label=below:\text{$1_0$}] {};
  \node[state] (b_0) at (4*\b, 0*\c) [label=below:\text{$2_0$}] {};
  \node[state] (b_1) at (4*\b, -2*\c) [label=above:\text{$2_1$}] {$\bullet$}; 
 \node[state] (c_0) at (5*\b, -4*\c) [label=right:\text{$3_0$}] {};
  \node[state] (c_1) at (3*\b, -4*\c) [label=left:\text{$3_1$}] {$\bullet$};
  
  \node[transition] (a_up) at (2*\b, -2.5*\c) {$1\UP$};
  \path[->] (a_0) edge [bend right=20] (a_up);
  \path[-] (c_0) edge (a_up);
  \path[->] (a_up) edge [bend right=20] (a_1);
  
 \node[transition] (a_down) at (0*\b, -2.5*\c) {$1\DW$};
 \path[->] (a_down)  edge [bend right=20] (a_0);
 \path[->]  (a_1) edge [bend right=20] (a_down);
 \path[-] (c_1) [bend left=75] edge (a_down);
 
  \node[transition] (b_down) at (5*\b, -1*\c) {$2\DW$};
   \path[->] (b_down)  edge [bend right=20] (b_0);
 \path[->]  (b_1) edge [bend right=20] (b_down);
 \path[-] (a_1) [bend left=75] edge (b_down);

  \node[transition] (b_up) at (3*\b, -1*\c) {$2\UP$};
    \path[->] (b_up)  edge [bend right=20] (b_1);
 \path[->]  (b_0) edge [bend right=20] (b_up);
 \path[-] (a_0) [bend right=50] edge (b_up);
 
  \node[transition] (c_down) at (4*\b, -5*\c) {$3\DW$};
    \path[->] (c_down)  edge [bend right=20] (c_0);
   \path[->] (c_1)edge [bend right=20] (c_down);
   \path[-] (b_1) [bend left=40] edge (c_down);
   
  \node[transition] (c_up) at (4*\b, -3*\c) {$3\UP$};
  \path[->] (c_up)  edge [bend right=20] (c_1);
   \path[->] (c_0)  edge [bend right=20] (c_up);
    \path[-] (b_0) [bend right=30] edge (c_up);
  
\end{tikzpicture}
\caption{\label{fig:ContextAneg} A translation of the BN
$\langle f_1(x) = \neg x_3, f_2(x) = \neg x_1, \neg f_3(x)\rangle
= x_3$ and the configuration $111$ into RPN.
The step $\stepn=\{1\DW,2\DW,3\DW\}$ is normal and reflects the negative-odd cycle of the BN.
}\end{center}
\end{figure}
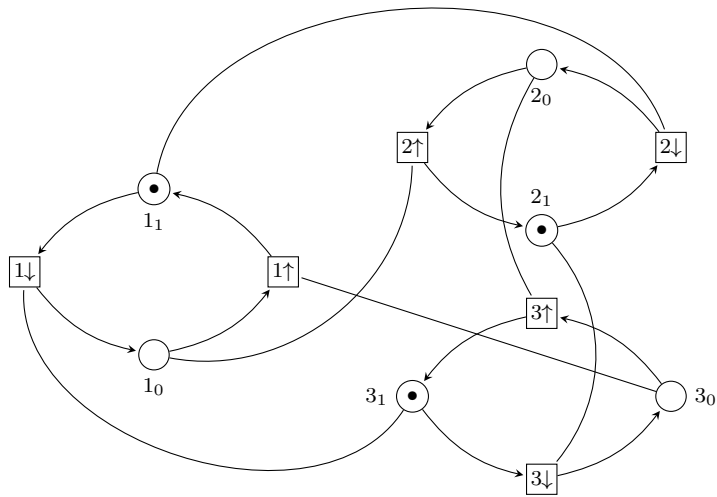
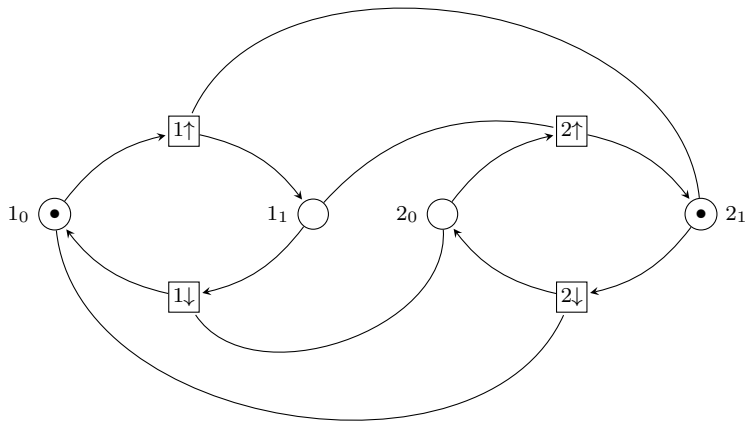
\begin{figure}
\begin{center}
\def\b{1.7}
\def\c{1.1}

\begin{tikzpicture}[>=stealth,shorten >=1pt,node distance=\c cm,auto]
  \node[state] (a_0) at (0*\b, -1*\c) [label=left:\text{$1_0$}] {$\bullet$};
  \node[state] (a_1) at (2*\b, -1*\c) [label=left:\text{$1_1$}] {};
  \node[state] (b_0) at (3*\b, -1*\c) [label=left:\text{$2_0$}] {};
  \node[state] (b_1) at (5*\b, -1*\c) [label=right:\text{$2_1$}] {$\bullet$}; 

  \node[transition] (a_up) at (1*\b, 0*\c) {$1\UP$};
  \path[->] (a_0) edge [bend left=20] (a_up);
  \path[-] (b_1) edge [bend right=75] (a_up);
  \path[->] (a_up) edge [bend left=20] (a_1);
  
 \node[transition] (a_down) at (1*\b, -2*\c) {$1\DW$};
 \path[->] (a_down)  edge [bend left=20] (a_0);
 \path[->]  (a_1) edge [bend left=20] (a_down);
 \path[-] (b_0) [bend left=75] edge (a_down);
 
  \node[transition] (b_up) at (4*\b, -0*\c) {$2\UP$};
    \path[->] (b_up)  edge [bend left=20] (b_1);
 \path[->]  (b_0) edge [bend left=20] (b_up);
 \path[-] (a_1) [bend left=30] edge (b_up);
 
  \node[transition] (b_down) at (4*\b, -2*\c) {$2\DW$};
   \path[->] (b_down)  edge [bend left=20] (b_0);
 \path[->]  (b_1) edge [bend left=20] (b_down);
 \path[-] (a_0) [bend right=75] edge (b_down);

\end{tikzpicture}
\caption{\label{fig:ContextB} A translation of the BN
$\langle f_1(x) = x_2, f_2(x)=x_1\rangle$
and configuration $01$ into RPN.
}\end{center}
\end{figure}

\subsection{Application to RPNs encoding BNs}

We now study how the characterization of synchronism sensitivity carries over to RPNs which
encode BNs following the transformation described in Sect.~\ref{sec:bn2cpn}.
Remember that in this setting, each transition $t$ of the RPN satisify $\preset
t=\{p\}$ and $\postset t=\{\non p\}$ with $\vari(p)=\vari(\non p)$ and
$\vali(p)+\vali(\non p)=1$.
Thus, $t$ corresponds either to an up-transition $\upt(\vari(p))$ iff $\vali(p)=0$ (i.e., $\vali(\non
p)=1$), or to a down-transition $\downt(\vari(p))$ iff $\vali(p)=1$ (i.e., $\vali(\non p)=0$).

Let us assume that the contexts of transitions are minimal, i.e., the DNF being the disjunction of all the context of all the
up- (resp. down-) transitions of a node is minimal.
Given an up-transition $t=\upt(\varn_i)$ (resp. a down-transition $t=\downt(\varn_i)$) of a node
$\varn_i$,
each place $p\in\context t$ corresponds to a node $\varn_j$ with $\vari(p)=\varn_j$.
Then, the sign of the influence from $\varn_j$ to $\varn_i$
is positive if $\vali(p)=1$ (resp. $\vali(p)=0$) and negative otherwise.

Consider  a preemption cycle $\transn_1\imasycf\ldots\imasycf\transn_n\imasycf\transn_1$, and any  arc $(\transn_i,\transn_{i+1})$, identifying $i=1$ and $i=n+1$ in this cycle.
By definition, there exists a place $p\in P$ with $\vari(p)=\varn_j=\vari(t_i)$ such that $\{p\} =
\preset{t_i}\cap\context{t_{i+1}}$,
and a place $q\in P$ with $\vari(q)=\varn_k=\vari(t_{i+1})$ and $\{q\}=\preset{t_{i+1}}$.
If $t_i=\upt(\varn_j)$ (i.e., $\vali(p)=0$),
and $t_{i+1}=\downt(\varn_k)$ (i.e, $\vali(q)=1$),
we say that the type of $(t_i,t_{i+1}$) is $0-1$, written $\nulleins$, and witnesses a positive influence
of $\vari(t_i)$ on $\vari(t_{i+1})$.
Similarly,
if $t_i=\downt(\varn_j)$ and $t_{i+1}=\upt(\varn_k)$,
the type of $(t_i,t_{i+1})$ is $1-0$, written $\einsnull$, witnessing a positive influence of
$\vari(t_i)$ on $\vari(t_{i+1})$;
if $t_i=\upt(\varn_j)$ and $t_{i+1}=\upt(\varn_k)$,
the type of $(t_i,t_{i+1})$ is $0-0$, written $\nullnull$, witnessing a negative influence of
$\vari(t_i)$ on $\vari(t_{i+1})$;
and if $t_i=\downt(\varn_j)$ and $t_{i+1}=\downt(\varn_k)$,
the type of $(t_i,t_{i+1})$ is $1-1$, written $\einseins$, witnessing a negative influence of
$\vari(t_i)$ on $\vari(t_{i+1})$.

As a consequence, in any preemption cycle, the number of type $\nulleins$ arcs and of type
$\einsnull$ arcs must be equal, while nothing can be said in general about the number of $\nullnull$ and
$\einseins$ arcs. Since  $\nullnull$ and $\einseins$ correspond to arcs with negative signs in the
BN's influence graph, adding them in a cycle does not change the cycle's NOPE status 
(it only changes from negative-odd to positive-even, or vice versa).

\bel
Let $\{\transn_1,\ldots,\transn_n\}$ be a preemption cycle in $\stepn$. Then the product of the signs of associated arcs $(\transn_i,\transn_{i+1})$ for $i\in\{1,\ldots,n-1\}$ and $(\transn_n,\transn_1)$ is positive iff $n$ is even.  
\enl

\begin{proof}
  By construction, the types of adjacent arcs have to match:  type 
  $\nulleins$ and  type 
  $\einseins$ arcs can only be followed by $\einsnull$ or $\einseins$, and analogously, types
  $\nullnull$ and $\einsnull$ need a successor arc of type $\nulleins$ or $\nullnull$.
  %
  Hence
  the word $w\in\{\nullnull,\nulleins,\einsnull\einseins\}^*$ associated to the preemption cycle must not contain the infixes $\nullnull\einsnull$, $\nulleins\nulleins$, $\nulleins\nullnull$, $\einsnull\einsnull$, $\einsnull\einseins$ or $\einseins\nulleins$, and not even $\nullnull\einseins$ or $\einseins\nullnull$. Since $w$ also has to be cyclic, this implies that 
  \bum
\item between any occurrences of $\nullnull$ and $\einseins$ ($\einseins$ and $\nullnull$), at least one occurrence of  $\nulleins$ ($\einsnull$) is required; 
\item between any two  occurrences of  $\nulleins$  ($\einsnull$), at least one occurrence of  $\einsnull$ ($\nulleins$) is required;
  \eum
  therefore $|w|_{\nulleins}=|w|_{\einsnull}$, which in turn implies the result.
\qed
\end{proof}

\begin{example}
The preemption cycle $\tau=\{1\DW,2\DW,3\DW\}$ in Fig.~\ref{fig:ContextAneg}
is of type $\einseins\einseins\einseins$, that of
$\tau'=\{1\UP,2\UP,3\UP\}$ of type $\nullnull\nullnull\nullnull$;
the preemption cycle $\tau''=\{1\UP,2\DW\}$ in Fig.~\ref{fig:ContextB} is of type
$\nulleins\einsnull$.
\end{example}

\section{Encoding the Interval Semantics with Boolean Networks}
\label{sec:bn-is}

In this section, we show how the interval semantics for RPNs
(Sect.~\ref{sec:interval-semantics}) can be modelled using BNs with
\emph{asynchronous} updating.
The resulting BNs subsume the generalized asynchronous updating mode, and enable
new reachable configurations, while preserving important dynamical and
structural (influence graph) properties.

The interval semantics relies on decomposing the firing of transitions in two stages: a first stage
checks the pre-conditions and commits the transition, and a second stage eventually applies the
transition (consuming and producing tokens).
Because of this decomposition, the interval semantics adds the possibility to trigger transitions
which become enabled during the firing of other transitions.
Essentially, its application to BNs can be modelled as follows.
Each node $i\in \nodes$ is decoupled in two nodes:
a ``write'' node storing the next value ($2i-1$)
and a ``read'' node for the current value ($2i$).
The decoupling is used to store an ongoing value change, while other nodes of the system still read
the current (to be changed) value of the node.
A value change is then performed according to the automaton given in Fig.~\ref{fig:automaton}:
assuming we start in both write and read node with value $0$,
if $f_i(x)$ is true, then the write node is updated to value $1$.
The read node is updated in a second step, leading to the value where both write and read nodes are
$1$.
Then, if $f_i(x)$ is false, the write node is updated first, followed, in a second stage by the
update of the read node.
\begin{figure}[t]
\centering
\begin{tikzpicture}[state/.style={circle,draw},node distance=2cm,alias/.style={font=\tt}]
\node[state,label={[alias]right:11}] (on) {1};
\node[state,label={[alias]left:01},below left of=on] (onoff) {\phantom 1};
\node[state,label={[alias]left:00},below right of=onoff] (off) {0};
\node[state,label={[alias]right:10},above right of=off] (offon) {\phantom 0};
\path[->,>=latex]
    (off) edge node[below,sloped] {$f_i(x)$} (offon)
    (offon) edge node[above,sloped] {$\epsilon$} (on)
    (on) edge node[above,sloped] {$\neg f_i(x)$} (onoff)
    (onoff) edge node[below,sloped] {$\epsilon$} (off)
;
\end{tikzpicture}
\caption{Automaton of the value change of a node $i$ in the interval semantics.
The states marked $0$ and $1$ represent the value $0$ and $1$ of the node.
The labels $f_i(x)$ and $\neg f_i(x)$ on edges are the conditions for firing the transitions;
$\epsilon$ indicates that the transitions can be done without condition.
The states are labeled by the corresponding values of nodes $(2i-1)(2i)$ in our encoding.}
\label{fig:automaton}
\end{figure}
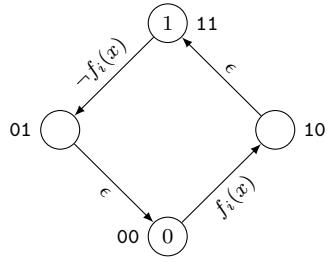

Once the write node ($2i-1$) has changed its value, it can no longer
revert back until the read node has been updated.
Hence, if $f_i(x)$ becomes false in the intermediate value $10$, the read node will still go
through value $1$ (possibly enabling transitions)
before the write node can be updated to $0$, if still applicable.

\subsection{Encoding}

From the automaton given in Fig.~\ref{fig:automaton}, one can derive Boolean functions for the write
($2i-1$)
and read ($2i$) nodes. It results in the following BN $\encode f$, encoding the
interval semantics for the BN $f$:

\begin{definition}[Interval semantics for Boolean networks]
\label{def:encode}
Given a BN $f$ of dimension $n$,
$\encode f$ is a BN of dimension $2n$ where
$\forall i\in \range n$,
\begin{align*}
\encode f_{2i-1}(z) &\DEF \left(f_i(\gamma(z)) \wedge (\neg z_{2i} \vee z_{2i-1})\right)
 \vee (\neg z_{2i} \wedge z_{2i-1})
\\
\encode f_{2i}(z) &\DEF z_{2i-1}
\end{align*}
where
$\gamma(z)\in\B^n$ is defined as
$\gamma(z)_i \DEF z_{2i}$ for every \(i\in \range n\).

\noindent Given $x\in\B^n$, $\alpha(x)\in \B^{2n}$ is defined as
$\alpha(x)_{2i-1} = \alpha(x)_{2i} \DEF x_i$ for every \(i\in\range n\).

\noindent A configuration $z\in\B^{2n}$ is called \emph{consistent} when
$\alpha(\gamma(z))=z$.
\end{definition}

\medskip
The function $\gamma:\B^{2n}\to\B^n$ maps a configuration of the interval
semantics to a configuration of the
BN $f$ by projecting on the read nodes.
The function $\alpha:\B^n\to\B^{2n}$ gives the interval semantics configuration
of a configuration of the Boolean
network $f$, where the read and write nodes have a consistent value.

The correctness of our encoding is given with respect to the interval semantics applied to the RPN
translation of the BN.
It follows from the correspondence between split transitions of the RPN and update of
read and write nodes of the encoded BN:
for any Petri net transition $t$ of the RPN $\cpnofbn f$,
the triggering of $t^-$ matches with the update of the ``write node'' for $\vari(\postset t)$ of the BN, and the
triggering of $t^+$ matches with the update of the ``read node'' for $\vari(\postset t)$ of the BN.
\begin{theorem}\label{thm:simulation}
Given a BN $f$ of dimension $n$,
for all $x,y\in \B^n$,
\[
\cpnofbn x\piite{\cpnofbn f}\closure\cpnofbn y
\Longleftrightarrow
\alpha(x)\areach{\encode f} \alpha(y)\enspace.
\]
\end{theorem}
%


\subsection{Consistency}
\label{sec:consistency}

The above theorem shows that the asynchronous updating of the BN $\encode f$ encoding the
interval semantics can reproduce any behaviour of the generalized asynchronous updating of $f$.
The aim of this section is to show that the interval semantics still preserves important constraints
of the BN on its dynamics.
In particular, we show the one-to-one relationship between the fixpoints of the BN and
its encoding for interval semantics; and that the influences are preserved with their sign.

Lemma~\ref{lem:consistent} states that from any configuration of encoded BN, one can always
reach a consistent configuration:
\begin{lemma}[Reachability of consistent configurations]
\label{lem:consistent}
For any $z\in \B^{2n}$ such that
$\alpha(\gamma(z)) \neq z$,
$\exists y\in \B^n: z\areach{\encode f} \alpha(y)$.
\end{lemma}
\begin{proof}
For each $i\in \nodes$ such that
$z_{2i-1}\neq z_{2i}$,
we update the $2i$ node, in arbitrary order.
This leads to the configuration $z'\in \B^{2n}$
where $\forall i\in\nodes$, $z'_{2i}=z'_{2i-1}=z_{2i-1}$.
Hence, by picking $y=\gamma(z)$, we obtain
$z\areach{\encode f} \alpha(y)$.
\qed
\end{proof}

The one-to-one relationship between fixpoints of $f$ and fixpoints of $\encode f$ is given by the
following lemma:
\begin{lemma}[Fixpoint equivalence]\label{lem:fixpoints}
$\forall x\in \B^n$,
$f(x)=x \Rightarrow f(\alpha(x))=\alpha(x)$;
and
$\forall z\in \B^{2n}$,
$\encode f(z)=z \Rightarrow \alpha(\gamma(z))=z \wedge
f(\gamma(z)) = \gamma(z)$.
\end{lemma}
\begin{proof}
Let $x\in \B^n$ be such that $f(x)=x$.
We have that
$\alpha(x)_{2i-1}=\alpha(x)_{2i}=x_i=f_i(x)$.
Hence,
$\encode f_{2i-1}(\alpha(x)) = f_i(\gamma(\alpha(x)))=f_i(x)=\alpha(x)_{2i-1}$;
and $\encode f_{2i}(\alpha(x)) = \alpha(x)_{2i-1} = \alpha(x)_{2i}$.
Thus, $\encode f(\alpha(x))=\alpha(x)$.

Let $z\in \B^{2n}$ be such that $\encode f(z)=z$.
For each $i\in \nodes$,
because $\encode f_{2i}(z)=z_{2i}$, by the definition of
$\encode f_{2i}$, we obtain that $z_{2i}=z_{2i-1}$.
Thus, $\alpha(\gamma(z))=z$.
Moreover,
as
$(\neg z_{2i}\vee z_{2i-1})$ reduces to true
and
$(\neg z_{2i}\wedge z_{2i-1})$ reduces to false,
$\encode f_{2i-1}(z) = f_i(\gamma(z)) = z_{2i-1} = \gamma(z)_i$.
Therefore, $f(\gamma(z))=\gamma(z)$.
\qed
\end{proof}

\subsection{Influence graph}
As defined in Sect.~\ref{sec:defs},
the influence graph provides a summary of the causal dependencies between the value changes of nodes
of the BN.
We show that our encoding of interval semantics preserves the causal dependencies of the original
network, and in particular, preserves the cycles and their signs.

From the definition of $\encode f$, one can derive
that all the influences in $f$ are preserved in $\encode f$, and no additional influences between
different variables $i,j$ are created by the encoding.
This latter fact is addressed by the following lemma:
\begin{lemma}\label{lem:influences}
For any $i,j\in\nodes$, $i\neq j$,
there is a positive (resp. negative) edge from $j$ to $i$
in $G(f)$ if and only if
there is a positive (resp. negative) edge
from $2j$ to $2i-1$ in $G(\encode f)$.
\end{lemma}
\begin{proof}
Let us define $x,y\in \B^n$ such that
$\Delta(x,y)=\{j\}$, and
$z,z'\in\B^{2n}$ such that
$z=\alpha(x)$
and
$\Delta(z,z')=\{2j\}$, i.e.,
$z'_{2j} = y_j$.
Because $z_{2i}=z_{2i-1}$
and, as $i\neq j$,
$z'_{2i}=z'_{2i-1}$,
we obtain that
$\encode f_{2i-1}(z) = f_i(x)$
and
$\encode f_{2i-1}(z') = f_i(y)$.
\qed
\end{proof}

\begin{lemma}For any $i\in\nodes$,
\label{lem:self-influences}
\begin{enumerate}
\renewcommand{\theenumi}{\alph{enumi}}
\item there is a positive self-loop on $2i-1$ in $G(\encode f)$ if and only if
there exists $x\in\B^n$ such that $f_i(x)=x_i$;
\item there is never a negative self-loop on $2i-1$ in $G(\encode f)$;
\item there is never a positive edge from $2i$ to $2i-1$ in $G(\encode f)$;
\item there is a negative edge from $2i$ to $2i-1$ in $G(\encode f)$ if and only
if there exists $x\in\B^n$ such that $f_i(x)\neq x_i$
\item there is always exactly one edge from $2i-1$ to $2i$ in $G(\encode f)$ and
it is positive.
\end{enumerate}
\end{lemma}
\begin{proof}
(a)
Let us consider $z,z'\in\B^{2n}$ such that $\Delta(z,z')=\{2i-1\}$ with
$z_{2i-1}=0$:
$\encode f_{2i-1}(z)=0=\neg\encode f_{2i-1}(z')
\Leftrightarrow
[(z_{2i}=0 \wedge f_i(\gamma(z))=0)
\vee
(z_{2i}=1 \wedge f_i(\gamma(z))=1)]
\Leftrightarrow
f_i(\gamma(z))=z_{2i}$.
(b)
Let us consider $z,z'\in\B^{2n}$ such that
$\Delta(z,z')=\{2i-1\}$ with $z_{2i-1}=0$
and
$\encode f_{2i-1}(z)= 1=\neg \encode f_{2i-1}(z')$.
Thus, $z_{2i}=0$, therefore,
$\encode f_{2i-1}(z')=z'_{2i-1}=1$, which is a contradiction.
(c)
Let us consider $z,z'\in\B^{2n}$ such that
$\Delta(z,z')=\{2i\}$ with $z_{2i}=0$:
if $z_{2i-1}=z'_{2i-1}=0$, then
$\encode f_{2i-1}(z) \geq \encode f_{2i-1}(z')$;
if $z_{2i-1}=z'_{2i-1}=1$, then
$\encode f_{2i-1}(z) \geq \encode f_{2i-1}(z')$;
therefore there cannot be a negative edge from $2i$ to $2i-1$ in $G(\encode f)$.
(d)
$\exists z,z'\in\B^{2n}$: $\Delta(z,z')=\{2i\}$,
$z_{2i}=0$,
$\encode f_{2i-1}(z)=1=\neg\encode f_{2i-1}(z')
\Leftrightarrow
[(z_{2i-1}=z'_{2i-1}=0 \wedge f_i(\gamma(z))=1)
\vee
(z_{2i-1}=z'_{2i-1}=1 \wedge f_i(\gamma(z'))=0)]
\Leftrightarrow
\exists x\in\B^n: f_i(x)=\neg x_i$.
(e) By $\encode f_{2i}$ definition.
\qed
\end{proof}

From Lemma~\ref{lem:self-influences}, one can deduce that
if there is a positive self-loop on $i$ in $G(f)$, then
there is a positive self-loop on $2i-1$ in $G(\encode f)$;
and if there is a negative self-loop on $i$ in $G(f)$, then
there is a negative edge from $2i$ to $2i-1$ in $G(\encode f)$.

We can then deduce that the positive and negative cycles of $G(f)$ are preserved in $G(\encode f)$.
It is worth noting that the encoding may also introduce negative cycles between $2i-1$ and $2i$
and positive self-loops on $2i-1$, for some $i\in\nodes$.
\begin{lemma}
To each positive (resp. negative) cycle in $G(f)$ of length $k>1$, there exists a corresponding
positive (resp. negative) cycle in $G(\encode f)$ of length $2k$.
To each positive self-loop in $G(f)$ corresponds one positive self-loop in
$G(\encode f)$;
to each negative self-loop in $G(f)$ corresponds a negative cycle in $G(\encode
f)$ of length $2$.
\end{lemma}
\begin{proof}
For cycle of length $k>1$, by Lemma~\ref{lem:influences} and by the fact that
there is a positive edge from $2i-1$ to $2i$ in $G(\encode f)$:
each edge $(i,j)$ in the cycle in $G(f)$ is mapped to the string
$(2i,2j-1)(2j-1,2j)$, giving a cycle in $G(\encode f)$ of the same sign.
Correspondence of self-loops is given by Lemma~\ref{lem:self-influences}.
\qed
\end{proof}

\section{Beyond Generalized Asynchronicity and Interval Semantics}
\label{sec:beyond-general}

BNs are widely used to model the qualitative dynamics of biological networks,
notably of signalling and gene regulation networks.

A major concern is the impact of the chosen updating mode on the validation of
the model.
Indeed, it is usual to assess the accordance of a BN with measurement data,
including time series: it is expected that the observed behaviours can be
reproduced in the abstract model.
With this perspective, the computation of reachable configurations in BNs is
key.
For example, let us assume we observe (in the concrete system) that a given
component (e.g., gene) gets eventually activated:
if the reachability analysis of the BN concludes that no reachable state has this component active,
the model would likely be rejected by the modeller.

In biological applications, the analysis of BNs merely splits into two
scientific sub-communities: the one preferring the synchronous updating mode,
and the one preferring the asynchronous updating mode.
The generalized asynchronous updating, which subsumes synchronous and
asynchronous, seems a good compromise but it received very little attention in
practice. It should be noted that most of computational tools rely only on either synchronous or
asynchronous modes, which can provide a partial explanation.

Is the generalized asynchronous mode the ultimate updating mode when
analysing reachable configurations in BNs for biological systems?
If little is known on time and speed features of the system and the reachability
analysis with generalized asynchronicity concludes on the absence
of the observed state, can we safely invalidate the model?

In the following motivating example (Sect.~\ref{sec:example}),
we show that the generalized asynchronous updating can miss transitions, hence
reachable configurations, which correspond to particular, but plausible,
behaviours.  Thus, the resulting analysis can be misleading on the absence of
some behaviours, notably regarding the reachability of attractors
(configurations reachable in the long-run), and may lead to rejection of valid models.
It is worth noting that the network considered in the example is embedded in many
actual models of biological networks, e.g., \cite{Mai09,Martinez13,Traynard16}.

As introduced in Sect.~\ref{sec:CPN}, the interval semantics of RPNs takes
advantage of the fine-grained specification of causality of transitions to
enable new behaviours, i.e., new reachable states, which can be caused by
specific ordering and duration of updates.
We show in Sect.~\ref{sec:apply-is} that using the encoding of BNs into RPNs provided in
Sect.~\ref{sec:bn2cpn} and applying the interval semantics correctly recovers the
missing reachable configurations in our motivating example.

Finally, in Sect.~\ref{sec:final-semantics} we explore further extensions of the interval semantics
resulting in correct over-approximation of the configurations reachable by any multi-valued
refinement of the BN.

\subsection{Motivating example}\label{sec:example}

Let us consider the BN defined in Fig.~\ref{fig:example}.
The BN and its influence graph suggest that the activity of species 3 increases when 1 is
inactive and 2 is active. In any scenario starting from \(000\) where 3
eventually increases, 2 has to increase to trigger the increase of 3. Hence, according to
the generalized asynchronous updating represented in Fig.~\ref{fig:example}
(c), the only transition which represents an increase of \(3\) is \(010
\rightarrow 011\). After this, no transition is possible.

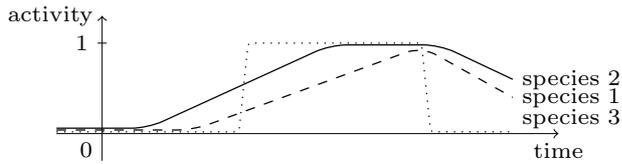
\begin{figure}[t]
\centering
\scalebox{1.2}{
  \begin{tikzpicture}[node distance=2cm]
    \draw[->,very thin] (-1,0) -- (-.5,0) node[below left]{\scriptsize 0} -- (4.5,0) node[below]{\scriptsize time};
    \draw[->,very thin] (-.5,-.3) -- (-.5,1.3) node[left]{\scriptsize activity};
    \draw[-,very thin] (-.55,1) node[left]{\scriptsize 1} -- (-.45,1);

    \draw[rounded corners=5pt] (-1,0.06) -- (0,0.06) -- (2,.98) -- (3.2,.98) --
    (4,.6) node[right]{\scriptsize species 2};
    \draw[rounded corners=5pt, dashed] (-1,0.04) -- (0.5,0.04) -- (3,.96) --
    (4,.4) node[right]{\scriptsize species 1};
    \draw[rounded corners=1pt, dotted] (-1,0.02) -- (1,0.02) -- (1.1,1) --
    (3.0,1) -- (3.1,0.02) -- (4,0.02) node[right, yshift=0.15cm]{\scriptsize species 3};
  \end{tikzpicture}}
\caption{A possible evolution of the activity of species modelled by the
BN of Fig.~\ref{fig:example} (species 1 in dashed line, species 2 plain, species
3 dotted).}
\label{fig:example-beyond}
\end{figure}

But, assuming the BN abstracts continuous evolution of activities,
the following scenario, pictured in Fig.~\ref{fig:example-beyond},
becomes possible: initially, the inactivity of species 1 causes an increase of the
activity of species 2, represented in plain line on the figure. Symmetrically,
the absence of species 2 causes an increase of the activities of species 1 (dashed
line). This corresponds to the evolution described by the arrow \(000
\rightarrow 110\) in Fig.~\ref{fig:example}(b) and leads to a (transient)
configuration where species 1 and 2 are present.

Assume that 1 and 2 activity increase slowly.
After some time, however, the activity of 2 becomes sufficient for
influencing positively the activity of 3, while there is still too little of
species 1 for influencing negatively the activity of 3. Species 3 can then increase.
In the scenario represented in the figure, 3 (dotted line)
increases quickly, and then 1 and 2 continue to increase. In summary, the
activity of species 3 increased from 0 to 1 \emph{during} the increase of 1 and
2, which was not predicted by the generalized asynchronous updating
(Fig.~\ref{fig:example}(b)).

One could argue that in this case, one should better consider more fine-grained
models, for instance by allowing more than binary values on nodes in order to
reflect the different activation thresholds.
However, the definition of the refined models would require additional
parameters (the different activation thresholds) which are unknown in general.
Our goal is to allow capturing these behaviours already in the Boolean
abstraction, so that any refinement would remove possible transitions, and not
create new ones.

\subsection{Application of the Interval Semantics of RPNs}
\label{sec:apply-is}

Let us consider the BN $f$ in Fig.~\ref{fig:example} and its RPN encoding
$\cpnofbn{f}$ in Fig.~\ref{fig:example2cpn}.
Starting from the marking $\cpnofbn{000}$,
$1\UP^-2\UP^-2\UP^+3\UP^-3\UP^+1\UP^+$
is a \emph{complete i-run} (Def.~\ref{def:complete-irun}) of the interval semantics,
and leads to the marking $\cpnofbn{111}$.

Similarly, let us consider the encoding of the interval semantics in the BN $\encode f$, as defined
in Sect.~\ref{sec:bn-is}.
We obtain the following possible sequence of
fully asynchronous updates of $\encode f$:
\begin{align*}
00\,00\,00\aite{\encode f}
10\,00\,00\aite{\encode f}
10\,10\,00\aite{\encode f}
10\,11\,00\\\aite{\encode f}
10\,11\,10\aite{\encode f}
10\,11\,11\aite{\encode f}
11\,11\,11
\end{align*}

Therefore, with the interval semantics, the configuration $111$ of $f$ is reachable from $000$, contrary to
the generalized asynchronous semantics.
This is due to the decoupling of the update of node $1$: the activation of $1$ is delayed which
allows activating node $3$ beforehand.

\subsection{Beyond the Interval Semantics}
\label{sec:final-semantics}

With the interval semantics, during the interleaving of transitions, the nodes have access only to
the before-update value of other nodes.
Moreover, the interval semantics enforces the update application: once an
update is triggered (write node gets a different value than the read node), no further update on the
same node is possible until the update has been applied.
Thus, if for instance the update triggers a change of value from $0$ to $1$, the interval
semantics guarantees that the read node will eventually have the value $1$.

In terms of modeling, the restriction to before-update values in our interval
semantics can be seen as an asymmetry in the consideration of transitions:
the resource modified by the transition is still available during the interval of update, whereas
the result is only available once the transition finished.
When modelling biological systems, it translates into considering only species
which are slow to reach their activity threshold.

Actually, the choice of whether the before-update, after-update or both values
are available during the update may be done according to the knowledge of the
modeled system.
Our construction in Sect.~\ref{sec:bn-is} can easily be adapted for giving access, depending on the node,
to the after-update value instead of the before-update value.
For instance, if the node $j$ should follow closely value changes of node $i$,
then node $j$ should access the after-update value (write node) of $i$, whereas, as in
our motivating example, if $i$ is slow to update compared to $j$,
node $j$ should access the before-update value (read node) of $i$.

\subsubsection{Most Permissive Fully Asynchronous Semantics for Boolean Networks}

Finally, we consider here a more permissive symmetric version which would allow
the access of both before-update and after-update values and do not enforce update application.
This choice may be very reasonable when not much is known about the system, for
instance about the relative speed of the nodes.

This leads us to define a \emph{most permissive fully asynchronous semantics} for BNs
which is defined as a 3-valued semantics in order to represent non instantaneous
updates: a component in a configuration can now have value \(\frac12\), in addition
to the usual 0 and 1, and the updates are done in two stages: if the network is
in a configuration \(x\) where for some \(i\), \(f_i(x)
\neq x_i\), the update of \(x_i\) will be in two stages, going through an
intermediate configuration \(y\) with \(y_i = \frac12\). In this intermediate
configuration \(y\), other updates can occur before the completion of the update
of node \(i\), and they will be allowed to use either the value 0 or 1 for node
\(i\). In the end, for a 3-valued configuration \(x \in \{0,\frac12,1\}^n\), we
allow all the intermediate values to be approximated either as 0 or as 1. The
possible approximations are defined as the set \(\Approx(x)\) of Boolean
configurations \(x' \in \B^n\) such that, for every \(i \in \range n\),
\begin{itemize}
\item \(x'_i = 0\) if \(x_i = 0\),
\item \(x'_i = 1\) if \(x_i = 1\),
\item otherwise \(x'_i\) can be either 0 or 1.
\end{itemize}

\begin{definition}[Most permissive fully asynchronous semantics for Boolean networks]
  Given a BN $f$, the binary irreflexive relation $\mpite f\,\subseteq \{0,\frac12,1\}^n\times\{0,\frac12,1\}^n$
  is defined as:
\begin{align*}
  x \mpite f y\EQDEF \exists &i\in \range n, x' \in \Approx(x): \Delta(x,y)=\{i\}
\\&
\wedge y_i=
\begin{cases}
f_i(x') & \text{if }x_i=\frac12\\
\frac12 &\text{otherwise }(x_i\neq f_i(x'))\enspace.
\end{cases}
\end{align*}
  We write $\mpreach f$ for the transitive closure of $\mpite f$.
\end{definition}

Similarly to the BN encoding of interval semantics presented in Sect.~\ref{sec:bn-is},
the most permissive fully asynchronous semantics of a BN $f$ of dimension $n$
can be encoded as an asynchronous BN $\encodemp f$ of dimension $3n$
where each node $i\in\range n$ is decoupled into an after-update value node $(2i-1)$
and a before-update value node $(2i)$.
As in Def.~\ref{def:encode}, the updating of this latter node consists in copying the after-update
value node: $\encodemp f_{2i}(z)\DEF z_{2i-1}$.
The definition of $\encodemp f_{2i-1}$ is a bit more complex as one has to rewrite $f_i(x)$ to
use (non-deterministically) either the before-update or after-update value of input nodes.
This non-deterministic choice can be encoded using extra ``coin flip'' nodes $(2n+j)$ for
$j\in\range n$ with $\encodemp f_{2n+j}(z)\DEF \neg z_{2n+j}$.
Then, assuming $f_i(x)$ is specified using propositional logic,
the literals $x_j$ appearing in $f_i(x)$ are replaced with
$\Tilde{\Tilde{x}}_j\DEF(z_{2n+j}\vee z_{2j})\wedge (\neg z_{2n+j}\vee z_{2j-1})$.
Also, contrary to the interval semantics, the most permissive fully asynchronous semantics does not enforce the update
application.
Thus, $\encodemp f_{2i-1}(z)\DEF [f_i(x)]_{[\Tilde{\Tilde{x}}_j/x_j, j\in\range n]}(z)$.

\subsubsection{Most permissive fully asynchronous semantics simulates any multivalued refinement}

Multivalued networks are generalization of BNs where the nodes
\(x_i\) can take values other than \(\{0, 1\}\). Let us denote the possible
values as \(\M \eqdef \{0, \frac1m, \dots, \frac{m-1}m, 1\}\) for some integer \(m\).
For simplicity in notations, we assume the same number of values for all the nodes.
A \emph{configuration} is now a vector \(x \in \M^n\).
Given two configurations $x,y\in \M^n$, the components that differ are noted
$\Delta(x,y)\DEF\{ i\in \range n\mid x_i\neq y_i\}$.

In practical modelling applications, multivalued networks enable considering different thresholds
for the interactions from one component to its regulators:
for instance, the activation of a second component may require the first component to be
only slightly active ($\frac 1m$), whereas the activation of a third component may require the full
activation ($1$) of this first one.

Hence, multivalued networks can be considered as refinements of BNs, where in addition to the
logic of interactions, one can mix different thresholds to consider a component active or inactive.
This fined-grained specification requires more information on the system, and it is then natural to
aim at performing analyses at a more abstract level (BN) and then transfering the results to possible
multivalued concretisations of the model.

In this section, we show that the most permissive fully asynchronous semantics enables such a
reasoning for reachability properties: essentially, this semantics captures any
behaviour possible in any multivalued refinement of the BN \emph{with asynchronous updating}.
Therefore, if a configuration is not reachable in the most permissive fully asynchronous semantics,
there exists no multivalued refinement for which the configuration become reachable with
asynchronous updating.

We illustrate this result with Examples~\ref{ex:mpa-1} and~\ref{ex:mpa-2} at the end of the section.
Notably, the last one shows an example where both generalized asynchronous updating and interval
semantics of a BN fail to capture behaviours which are actually possible in a
multivalued refinement of it; these behaviours are correctly preserved by the most permissive fully
asynchronous semantics.

\medskip

From a specification point of view, multivalued networks can be defined similarly to BNs, except that
the functions now map the configurations to either "$\uparrow$" (increase the value of component by $\frac 1m$),
"$-$" (do not change the value of component), or "$\downarrow$" (decrease the value of component by $\frac 1m$).

Def.~\ref{def:multivalued-refinement} formalizes the notion of multivalued refinement:
a multivalued network $F$ refines a BN $f$
if, for every component $i\in\range n$,
for each multivalued configuration $x$,
if $F_i(x)$ leads to an increase (resp. decrease) of the value of $i$,
there is a binarization $x'\in\B^n$ of $x$ such that
$f_i(x')=1$ (resp. $f_i(x')=0$).
Here, the binarization allows to map non-binary values to either $0$ or $1$.

Theorem~\ref{thm:mpa} states that, given a BN $f$,
any \emph{fully asynchronous} transition of any multivalued refinement $F$ of $f$
is captured by the most permissive fully asynchronous semantics, possibly by the mean of several
intermediate transitions.

\begin{definition}[Multivalued network]
A \emph{multivalued network} of dimension $n$ over a value range \(\M =
\{0, \frac1m, \dots, \frac{m-1}m, 1\}\) is a collection of functions
$F=\langle F_1, \ldots, F_n\rangle$ where
$\forall i\in\range n, F_i:\M^n\to \{{\uparrow}, -, {\downarrow}\}$.
\end{definition}

\begin{definition}[Asynchronous updating in multivalued networks]
Given a multivalued network $F$, the binary irreflexive relation $\aite F\,\subseteq \M^n\times\M^n$
is defined as:
\begin{align*}
x\aite F y\EQDEF &\exists i\in \range n: \Delta(x,y)=\{i\} 
\\&\qquad\wedge y_i =
\left\{\begin{array}{@{}l@{\mbox{ if }}l@{}}
  \min\{0, x_i - \frac1m\} & F_i(x) = {\downarrow}\\
  \max\{1, x_i + \frac1m\} & F_i(x) = {\uparrow} \enspace.\\
\end{array}\right.
\end{align*}
We write $\areach F$ for the transitive closure of $\aite F$.
\end{definition}

We now define a notion of \emph{multivalued refinement} of a BN,
which formalizes the intuition that the moves defined by the multivalued network
are compatible with those of the BN.

\begin{definition}[Multivalued refinement]
  \label{def:multivalued-refinement}
  A multivalued network \(F\) of dimension $n$ over a value range \(\M =
  \{0, \frac1m, \dots, \frac{m-1}m, 1\}\) \emph{refines} a BN
  \(f\) of equal dimension \(n\) iff for every configuration \(x \in \M^n\)
  and every \(i \in \range n\):
  \begin{itemize}
  \item \(F_i(x) = {\uparrow} \implies \exists x' \in \Approx(x): f_i(x') = 1\)
  \item \(F_i(x) = {\downarrow} \implies \exists x' \in \Approx(x): f_i(x') = 0\)
  \end{itemize}
  where \(\Approx\) is generalized to multi-valued networks by \(\Approx(x)
  \eqdef \Approx(\abstr(x))\)
    with $\abstr: \M^n\to \{0,\frac12,1\}^n$ mapping every configuration of the multivalued network
    into a 3-valued configuration,
    which is defined for every $i\in\range n$ as:
  \begin{itemize}
  \item \(\abstr(x)_i \eqdef 0\) if \(x_i = 0\),
  \item \(\abstr(x)_i \eqdef 1\) if \(x_i = 1\),
  \item \(\abstr(x)_i \eqdef \frac12\) otherwise.
  \end{itemize}
\end{definition}


\begin{theorem}[Most permissive fully asynchronous semantics simulates any multivalued refinement]
    \label{thm:mpa}
  Let \(f\) be a BN of dimension $n$ and \(F\) a multivalued
  refinement of \(f\). Then
  \[\forall x, y \in \M^n, \quad
  x \aite{F} y \implies \abstr(x) \mpreach{f} \abstr(y)\,.\]
\end{theorem}
\begin{proof}
  We assume first that \(m > 1\).
  By definition of $\aite F$ for multivalued networks, there exists a unique
  \(i\) such that \(\Delta(x,y) = \{i\}\).
  Then we have to study the different cases determined by the value of \(x_i\)
  and of \(F_i(x)\).

  The first case is \(0 < x_i < \frac{m-1}m\) and \(F_i(x) = {\uparrow}\). It
  implies \(y_i = x_i + \frac1m\), and we observe that, in this case,
  \(\abstr(x) = \abstr(y)\). Then trivially \(\abstr(x) \mpreach{f}
  \abstr(y)\).
  The case of \(\frac1m < x_i < 1\) and \(F_i(x) = {\downarrow}\) is
  symmetric.

  The other cases are all similar; consider for instance \(x_i = 0\) and
  \(F_i(x) = {\uparrow}\), which imposes \(y_i = \frac1m\). Notice first that
  \(\Delta(\abstr(x), \abstr(y)) = \{i\}\) and \(\abstr(x)_i = 0\) and
  \(\abstr(y)_i = \frac12\). Now, since \(F\) is a multivalued refinement of
  \(f\), then by Def.~\ref{def:multivalued-refinement}, there exists an
  \(x' \in \Approx(x) = \Approx(\abstr(x))\) such that \(f_i(x') = 1\). 
  Thus, we get \(\abstr(x) \mpite{f} \abstr(y)\).
  The case when $x_i=1$ and $F_i(x)={\downarrow}$ is similar.
  Regarding the case when $y_i=1$ and $F_i(x)={\uparrow}$,
  note that $\abstr(x)_i=\frac12$ and $\abstr(y)_i=1$ and
  that there exists an \(x' \in \Approx(x) = \Approx(\abstr(x))\) such that \(f_i(x') = 1\).
  Thus, $\abstr(x)\mpite f\abstr(y)$.
  The case when $y_i=0$ and $F_i(y)={\downarrow}$ is similar.

  Finally, for \(m = 1\), consider the case where \(x_i = 0\) and
  \(F_i(x) = {\uparrow}\), which imposes \(y_i = 1\).
  Now \(\abstr(x) = x\) and \(\abstr(y) = y\). In the most permissive fully asynchronous
  semantics, we have \(x\mpite f z\mpite f y\) with an intermediate state \(z\)
  defined by \(\Delta(x, z) = \{i\}\) and \(z_i = \frac12\). The transition
  \(z\mpite f y\) is allowed because \(x \in \Approx(z)\).
  \qed
\end{proof}

\begin{example}
    \label{ex:mpa-1}
The scenario pictured in Fig.~\ref{fig:example-beyond} can be obtained as a
behaviour of a 3-level refinement $F$ of the BN $f$ in Fig.~\ref{fig:example}, with the following update functions:
\begin{align*}
F_1(x) &\DEF {\uparrow} \mathrm{\ if\ } x_2 < 1 \mathrm{\ else\ } \downarrow\\
F_2(x) &\DEF {\uparrow} \mathrm{\ if\ } x_1 < 1 \mathrm{\ else\ } \downarrow\\
F_3(x) &\DEF {\uparrow} \mathrm{\ if\ } x_1 \leq {\frac12} \wedge x_2 \geq {\frac12 \mathrm{\ else\ } \downarrow}
\end{align*}
We get
\(000 \aite F 0{\frac12}0 \aite F {\frac12}{\frac12}0 \aite F {\frac12}{\frac12}{\frac12}
\aite F {\frac12}{\frac12}1 \dots\).

In particular, imagine that a fourth species would activate when \(x_1\),
\(x_2\) and \(x_3\) are all \(\geq \frac12\), then even the generalized
asynchronous updating mode would not capture its activation, contrary to our
interval semantics for BNs.
\end{example}

\begin{example}
    \label{ex:mpa-2}
Let us consider the BN $f$ of dimension $3$ defined as follows:
\begin{align*}
f_1(x) &\DEF 1\\
f_2(x) &\DEF x_1\\
f_3(x) &\DEF x_2\wedge \neg x_1
\end{align*}
Starting from configuration $000$ the generalized asynchronous mode allows only the following
transitions $000\ite f 100\ite f 110$, where $110$ is a fixpoint of $f$.
The interval semantics lead to a very similar behaviour, with the following unique sequence of
asynchronous transitions of the BN encoding of the interval semantics:
\begin{align*}
00\,00\,00\aite{\encode f}
10\,00\,00\aite{\encode f}
11\,00\,00\aite{\encode f}
11\,10\,00\aite{\encode f}
11\,11\,00
\end{align*}
Indeed, in order to activate species $2$, $1$ has to be activated first as in the interval
semantics species $2$ only has access to the before-update value of $1$.
Then, once species $1$ is active, it is impossible to activate species $3$.

Now, let us consider the following 3-level refinement $F$ of the BN $f$:
\begin{align*}
F_1(x) &\DEF {\uparrow}\\
F_2(x) &\DEF {\uparrow}\IF x_1\geq \frac12\ELSE \downarrow\\
F_3(x) &\DEF {\uparrow}\IF x_2\geq\frac12\wedge x_1\leq\frac12\ELSE\downarrow
\end{align*}
The following asynchronous transitions are possible from configuration $000$:
$000\aite{F}\frac1200
\aite F\frac12\frac120
\aite F\frac12\frac12\frac12$.
These transitions are also transitions of the most permissive fully asynchronous semantics of $f$, $\mpite f$.
Essentially, as in this semantics species can have access to either the before-update or after-update
value of other species, species $2$ can be activated by reading the after-update value of $1$, while
species $3$ can be activated by reading the before-update value of $1$.
An example of possible sequence of asynchronous transitions of the BN encoding of the most permissive fully asynchronous
semantics is the following:
\begin{align*}
00\,00\,00\aite{\encodemp f}
10\,00\,00\aite{\encodemp f}
10\,10\,00 \aite{\encodemp f}
10\,11\,00 \\\aite{\encodemp f}
10\,11\,10\aite{\encodemp f}
10\,11\,11
\end{align*}

As in the previous example, let us consider a fourth species activated when $x_1$, $x_2$, and $x_3$ are all
greater or equal than $\frac12$: such an activation is captured neither by the generalized
asynchronous updating nor by the interval semantics of the abstract BN $f$, whereas it is captured by
its most permissive fully asynchronous semantics.
\end{example}

\section{Discussion}
\label{sec:discuss}

With this paper, we detailed the link between Boolean Networks (BNs) and Read (or contextual) Petri Nets
(RPNs) by focusing on the analysis of concurrency enabled by the latter framework.
On the one hand, BNs have prominent structural properties between the components and their
evolution, while on the other hand RPNs bring a fine-grained specification of the causality
and effect of transitions.
We show how we can take benefit of both approaches to first bring new updating modes to BNs by
encoding RPN semantics, and, secondly, propose further extensions of these semantics aiming at obtaining
correct Boolean abstractions of discrete dynamical systems.

To sum up, the contributions of this paper include:
\begin{itemize}
\item The encoding of BNs into RPNs, similar to other encodings already existing in the literature,
    here specialized for Read  Petri nets;
\item The encoding of RPNs into BNs, which allows a brief proof by reduction of the
PSPACE-completeness of the reachability decision in asynchronous BNs;
\item A generic characterization of synchronism sensitivity in RPNs, which when instantiated to BN
translations, allows to recover a recent result in BNs;
\item The encoding of the interval semantics of RPNs as asynchronous BNs, enabling new behaviours
missed by usual BN updating modes;
\item An extension of the interval semantics for BNs which guarantees to include the
behaviour of any multivalued refinement.
\end{itemize}

For practical applications, the thorough link between BNs and Petri nets enables the use of
conceptual tools 
 based on causality and concurrency, such as unfoldings
offering more compact representation of behaviours \cite{Esparza08,BBCKRS12,godunf-CONCUR,PRNs-TCS18} 
and for which efficient software tools have been developed for safe PNs \cite{mole} and RPNs
\cite{rodriguez2013cunf}.
For example, \cite{CHJPS14-CMSB,godunf-CONCUR} show the applicability of unfoldings to
analyse reachable states and attractors in BNs with biological use cases having up to 88
components.

The transitions enabled by the interval and most permissive semantics are due to nodes which update
at different time scales.
For instance with the interval semantics, whenever committed to a value change, in the meantime of the update
application, the other nodes of the network still evolve subject to its before-update value.
This time scale consideration brings an interesting feature when modeling biological networks which
gathers processes of different nature and velocity.
Our encodings can be applied only to a subset of nodes, offering a flexible modelling approach.
Moreover, because the encodings rely on asynchronous BNs, they can be implemented
using any software tools supporting the asynchronous updating mode.

The introduction of the most permissive fully asynchronous semantics for BNs motivates future work
to determine if it offers the smallest abstraction of any multivalued refinement (i.e., to any
transition of the most permissive
semantics corresponds an asynchronous transition of a multivalued refinement),
and to assess the complexity of reachability decision.
Finally, further work may 
explore links between BNs and RPNs with real-time semantics
\cite{DBLP:journals/fmsd/BalaguerCH12}, aiming at tightening connections between the two hybrid frameworks.

\section*{Acknowledgements}
The authors acknowledge the support from the French Agence Nationale pour la
Recherche (ANR), in the context of the
ANR-FNR project ``AlgoReCell'' ANR-16-CE12-0034,
from the Labex DigiCosme (project 
ANR-11-LABEX-0045-DIGICOSME) operated by ANR as part of the program 
``Investissement d'Avenir'' Idex Paris-Saclay (ANR-11-IDEX-0003-02),
from Paris Ile-de-France Region (DIM RFSI),
and  from UMI 2000 ReLaX (CNRS, Univ. Bordeaux, ENS Paris-Saclay, CMI, IMSc) for the internship of
Aalok Thakkar at ENS Paris-Saclay, at that time student at Chennai Mathematical Institute, India,
and during which part of this work was done.



\bibliographystyle{spmpsci}      
\bibliography{biblio}   


\end{document}